%% file: mangocdn.tex
\documentclass[conference]{IEEEtran}
\makeatletter
\def\ps@headings{%
\def\@oddhead{\mbox{}\scriptsize\rightmark \hfil \thepage}%
\def\@evenhead{\scriptsize\thepage \hfil \leftmark\mbox{}}%
\def\@oddfoot{}%
\def\@evenfoot{}}
\makeatother
\newcommand{\eat}[1]{}
\newcommand{\cut}[1]{}
\newcommand{\comm}[1]{{\bf [[ #1 ]]}}

\usepackage{amsmath}
\usepackage{graphicx}
\usepackage{subfigure}
\usepackage{algorithm , algorithmic}

\newtheorem{thm}{Theorem}[section]
\newtheorem{lemma}{Lemma}[section]

\newcommand{\etal}{et.al.}
\newcommand{\bytecost}{byte-cost}
\newcommand{\round}{\textsc{Round}}
\newcommand{\caching}{\textsc{DataPlacement}}

\newcommand{\clust}{\textsc{Cluster}}

\title{Bulk content delivery using co-operating end-nodes with upload/download limits}

\author{Sharad Jaiswal, Anirban Majumder, K V M Naidu, Nisheeth Shrivastava\\ Bell Labs Research India, Bangalore, India.}

\date{}

\begin{document}

\maketitle
\begin{abstract}

We study the problem of optimizing the cost of content delivery in a cooperative network of caches at end-nodes. The caches could be, for example, within the computers of users downloading videos from websites (such as Netflix, Blockbuster etc.), DVRs (such as TiVo, or cable boxes) used as part of video on demand services or public hot-spots (e.g. Wi-Fi access points with a cache) deployed over a city to serve content to mobile users. 
Each cache serves user requests locally over a medium that incurs no additional costs (i.e. WiFi, home LAN); if a request is not cached, it must be fetched from another cache or a central server. In our model, each cache has a tiered back-haul internet connection, with a usage cap (and fixed per-byte costs thereafter). Redirecting requests intended for the central server to other caches with unused back-haul capacity can bring down the network costs. Our goal is to develop a mechanism to optimally 1) place data into the caches and 2) route requests to caches to reduce the overall cost of content delivery. %


We develop a multi-criteria approximation based on a LP rounding procedure that with a small (constant factor) blow-up in storage and upload limits of each cache, gives a data placement that is within constant factor of the optimum. Further, to speed up the solution, we propose a technique to cluster caches into groups, solve the data placement problem within a group, and combine the results in the rounding phase to get the global solution.
Based on extensive simulations, we show that our schemes perform very well in practice, giving costs within $5\--15$\% to the optimal, and reducing the network load at a central server by as much as $55$\% with only a marginal blow up in the limits. Also we demonstrate that our approach out-performs a non-cooperative caching mechanism by about $20\%$.
\end{abstract}



\vspace{-.05in}

\section{Introduction}\label{sec:intro}

The internet has seen a massive growth in traffic in recent years, and
CDNs (Content Distribution Networks) play an important role in increasing
the effective network bandwidth and reducing content
access latency. A significant percentage of internet traffic is now driven by
{\em bulk traffic}, i.e. large-sized files, such as video, software patches,
software downloads etc. In this work we explore the design of a CDN - incorporating 
co-operating end-users - for bulk content delivery. 

As a representative example of this class of traffic - we will motivate
our work by focussing on video downloads - driven by the observation that popular
video-sites are  amongst the top (and fastest growing) contributors
to overall Internet traffic.

Video applications are inherently bandwidth hungry, and with the increasing
popularity of full-length, long duration video downloads -  the
management of video traffic has become an important problem in today's
internet. For example, consider a web-site (such as  Netflix)
from which users download
and rent or purchase full-length movies.
Each movie file will have a size of about tens of GBs, and
assuming hundreds of thousands of active users, one can easily
envision the web-site and the network having to serve peta-bytes of data (and increasing) monthly.

A natural direction towards alleviating this problem is to involve end-users in
the content delivery architecture. Web-sites (and the network) can incentivize
and manage users to act as content relays, i.e. requests can be
redirected by the web-site to be served by another user who has a cached copy of the request.
Such p2p based systems not only reduce the load on the content servers, but as
has been pointed out by some recent work, they can also reduce
load on the network (and enhance end-user performance) by exploiting the geographical
locality of content~\cite{Zink:2009} (and careful peer selection). There
exist several practical examples of this approach - web-sites
that push out large software updates (e.g. Windows patches) to
millions of users~\cite{Abiteboul07,Gkantsidis},
and BBC's IPlayer\footnote{www.bbc.co.uk/iplayer/} video service have
employed p2p based techniques to propagate content across their users.

Video files naturally lend themselves to a peer-peer based approach for
content delivery. Video content is typically held in the end-user cache for
an extended period of time (for later/multiple viewings or simply as part of the
users library of content). Also - in case of systems such as TiVo or
video-on-demand services, the ISP or the content provider directly manage the
hardware cache. For these reasons, a carefully managed system (with appropriate
incentives for end-users) through which locally cached
content can be used as a content source for other users seems entirely
feasible and appropriate - in terms of potential gains.

Managed peer-peer content delivery has been a studied problem in literature
e.g. Telco-managed p2p TV~\cite{Meeyoung08},
CDN with managed {\it swarms}~\cite{Peterson:2009},
ISP managed p2p systems~\cite{James:2009}. 
With the rise of data hungry mobile devices (e.g. IPhone users), there is increased importance of Wi-Fi hot-spots used to ``side-load" content on to the user device, faster and cheaper than over the cellular data network. Going forward - such Wi-Fi hot-spots may have the ability to locally cache (to reduce back-haul costs to the hot-spots) and serve rich content. Hence, we believe, another important use-case is that of content delivery over a network of ``info-stations" at the edge of the network. Such a content-delivery model has also been explored earlier by several works, e.g. the Rutgers Infostations project~\cite{Infostations}, the Drive-thru-Internet~\cite{Ott04} and MIT CarTel~\cite{1161097} projects (in contexts such as opportunistic content-delivery to moving vehicle), and for low-cost rich content delivery on mobiles (e.g. the mango project~\cite{mango09}).

A managed p2p system of the kind we consider will
comprise of a set of peer caches and one or more content servers. The end-user connects
to one of the caches for content download. If the data is cached locally, it will
be served immediately, free-of-cost. Otherwise, the request will be forwarded to
the content server which will figure out the cache that retains a copy
of the content and the download will resume in a p2p fashion.

Once the end-nodes are involved in the caching process, the content
server needs to be even more mindful of the local constraints
while making the global decisions e.g. constraints like node failures,
storage capacity at the end-nodes become critical for survivability of the
network. These issues have been studied throughly in the literature of
delay tolerant networks\footnote{http://en.wikipedia.org/wiki/Delay-tolerant\_networking}.
In this work, we highlight another important aspect of end-node caching
that has been previously ignored in the literature --
given the deployment location of the end-nodes (homes, cafe, shops, malls etc.),
it is likely that the back-haul connection will be a {\it commodity broadband plan}.
Below we highlight some important facts regarding the commodity broadband
plans that are of use today --

\begin{list}{\labelitemi}{\leftmargin=0.5em}
\item A significant number of broadband plans are defined by a
{\it tiered} cost model i.e. the end user pays a fixed amount
every month with a usage limit, typically in the order of GBs. When this
capacity is exceeded, the user will pay a {\it penalty} in terms of per-byte
charge. Such a price model has been introduced/considered by US-based cable
broadband providers like Comcast, Time Warner~\cite{GigaOm}.
Similarly in India, major ISPs like Airtel, BSNL have a range of
tiered plans, with usage caps ranging from 1GB to 75GB.
Finally, such limits and rate plans will be the norm for the increasing popular {\em mobile} broadband connections.
\item In practical deployments, the content server has a high data rate connection that
is bought in bulk from the service provider and hence has typically a lesser penalty
than the content caches.
\end{list}

As per the facts mentioned above, the end-nodes are likely to have a
monthly limits on the amount of uploads and downloads that can be
allowed {\it free-of-cost}\footnote{With a fixed monthly bill.}.
In practice, not every node will reach this limit.
This scenario gives rise to an
interesting data management problem where a {\it missed} request for
a cache can be opportunistically routed to another one that has a spare capacity to
serve it, instead of always fetching it from the central server.
In this paper, we formally model a system of cooperative caches
with {\it tiered} cost function.
The unique cost structure that we explore in this paper adds a novel
angle to the existing work on data placement problem. Our cost
structure can be thought of as being \emph{stateful}; i.e. the cost of serving
a request depends on whether the cache has exceeded its limits (upload and download) or
not (cost is positive if the limit is exceeded, else zero) which, in turn, depends
on the set of requests served by the cache till that point of time. This is in stark
contrast with the {\em stateless} model where we are given a flat cost optimization
problem which has a constant cost per serving, irrespective of the download/upload limits.
We believe that the new cost model we consider in this paper, adds an intrinsic difficulty
in the caching problem.

\subsection{Contributions}

In this paper we formally model a system for
sideloading mobile content, and we make the following contributions:

\begin{list}{\labelitemi}{\leftmargin=0.5em}

 \item To the best of or knowledge,we are the first to model the back-haul
costs of the caches as a {\em tiered} cost function
(as is the case in practice), and show that
it gives rise to a novel problem previously unstudied in literature.
 \item We prove the problem to be NP-hard, and propose an algorithm for (near) optimum
placement of data on the caches. Our technique employs both a static as well as a dynamic
component. Based on the request pattern, we statically place data on the caches
to minimize the serving cost and supplement it with a cache replacement
strategy to adapt to a dynamically changing stream of requests.
\cut{that aims to achieve the lowest cost of placing and serving data from the system of caches, over a stream of requests dynamically changing in time.
}
 \item We present a  multi-criteria approximation algorithm for this problem, based on a
novel LP rounding technique. Our algorithm produces a data placement that is
within a constant factor of the optimum solution and results in a small (constant factor)
blow-up in stroage and upload limits of each cache. In experiments, our schemes give costs
close to the optimum (within $5\--15$\%), with a marginal blow up ($< 1$\%) of the limits.
\cut{
 that has a small (constant factor) blow-up in storage
and upload limits of each cache and gives a data placement within a
constant factor of the optimum. In experiments, our schemes give costs
close to the optimum (within $5\--15$\%), with a marginal blow up in the limits.
}
 \item Further, to speed up the solution we cluster caches into groups,
solve the data placement problem within a group, and combine the results in the
rounding phase to get the global solution. Experiments suggest that computationally
tractable cluster sizes come well within $10$\% of a best possible solution.
\cut{
 \item To further reduce computational time, we also develop a fast heuristic that greedily places the data, while staying within storage and upload/download limits. The heuristics on average provides solutions $10\--30$\% more expensive than the rounding approach described above, but achieves very fast execution times.
}
 \item  We demonstrate that co-operative caching with the tiered network cost model
brings down the network delivery at a central server by as much as $55\%$,
and is $20\%$ better than a non-cooperative strategy.

\end{list}

\cut{
Distributed, co-operative content-delivery has been very successfully
exploited by peer-peer applications (such as BitTorrent). However,
the model under which they operate is different \-- there is no
centralized control, and any management of the end-node resources is purely
local (e.g. cap outgoing BitTorrent sessions to 32Kbps). In the model
we propose, a central entity collects the end-node information, and while
mindful of the local constraints, makes optimum decision for the entire system.
The notion of co-operative caching has also been studied extensively in the
literature before in the context of co-operating web proxies, or content
distribution networks (such as Akamai). However, in our work, we assume the
caches to be end-nodes, and the accompanying model
for network bandwidth costs (tiered broadband model) is very different
from the cost model assumed
in previous studies (typically, per byte costs). As we will explain
in detail later, this tiered cost model (particularly relevant
for caching on end-nodes) creates a fundamentally different problem.
}

\section{RoadMap}
The paper is organized as follows. In Section~\ref{sec:model} we formally
describe our assumptions and the system model and formulate the resulting
optimization problem. We then survey the related work in this
area in Section~\ref{sec:related}. Section~\ref{sec:algos} presents
algorithms for a simplified data placement problem and analyze their properties
and performance. In Section~\ref{sec:online} we present the complete solution
for the data placement problem.
\cut{In Section~\ref{sec:ext} we present some extensions to bring down the computational time of attaining the solution and also compute the initial cost of data placement. }
We follow it up with a rigorous evaluation through simulations in
Section~\ref{sec:simulations}. Finally we conclude with a summary of the
work and future directions in Section~\ref{sec:conclusions}.
\section{System Model and Problem Formulation}
\label{sec:model}

\begin{figure}[t]
\begin{center}
\includegraphics[width=.32\textwidth] {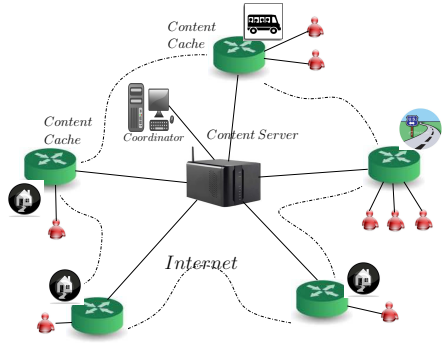}
\end{center}
\vspace{-.2in}
\caption{A system of distributed, co-operating caches}
\vspace{-.2in}
\label{fig:sysmodel}
\end{figure}

Our system of co-operating caches consists of three components
- \emph{content caches} at the end-nodes, a \emph{content server}
and a \emph{coordinator}.  Users are associated with one or more content
caches, and can either access it directly (the cache is a disk on the
user's home computer), or via Ethernet LAN, WiFi or Bluetooth (the
cache as a home networked storage device, or a public hot-spot).

The \emph{content server} is a repository of all items requested, uploaded and
shared by users. Any request at a cache is either served locally, or via
the content server, or by any other content cache. The {\em coordinator}
also keeps track of the user behavior and the state of each cache (items stored,
spare storage and capacity on the back-haul links).

A user's content viewing habit will consist of sites visited
regularly for popular content - e.g. online newspapers,
video sites (such as TV episodes on Hulu), or recently popular user
generated content on YouTube. In general, we believe there will be a strong
predictive element to
the popularly requested items in the system, and our system is
designed to exploit this characteristic.

Based on usage statistics at each cache,
sharing and mobility patterns of users, popularity of content items etc., periodically
(over a time frame referred to as an \emph{epoch}) the coordinator
{\em predicts} the demands at each cache. Then, based on these predictions in each epoch,
it i) places the items at each cache and
ii) redirects subsequent user requests to suitable content caches, or to the content server.

The specific techniques required to predict user requests are
an active and important area of research (and has been pursued
in other works~\cite{youtube-09,pred-kdd07}), but is orthogonal and beyond the
scope and goal of this paper.\footnote{However we would like to re-iterate there will be a naturally strong predictive element to the request
workload in our target applications. For example, video web-sites such as YouTube, Netflix etc. already compute and maintain detailed statistics on what is currently popular - and these statistics remain stable across the time-scale of decision-making assumed in our algorithms\cite{Cha:2007}.}

Also, typically, items popular in an epoch will retain their popularity for 
a period of time\cite{Cha:2007}. This can improve caching performance as 
the existing set of items in the cache need not be brought from 
the server in subsequent epochs. Hence, our system is designed to make data 
placement decisions keeping in mind the items already present in caches from previous epochs.

\subsection{Problem Formulation}\label{sec:prob}

\cut{
If a request
can't be served locally, the corresponding \emph{Content Cache} will access
the coordinator to identify the best location to download it from.
}

Consider a network of $n+1$ nodes $C=\{C_{0\cdots n}\}$.
Each node $C_i$ ($i > 0$) represents a content cache.
$C_0$ represents the content server 
\cut{Each node is connected to the
Internet, which allows it to send or receive data from any
other node in the network.
}
For cache $C_i$, let $u_i$ and $d_i$ specify the limit on data that
can be uploaded/downloaded at the cache \emph{free-of-cost}
\footnote{i.e. beyond the fixed monthly bill.}. For any data
transfer beyond the prescribed limits, $\alpha_i$ and $\beta_i$ denote
the price \emph{per-byte} paid by the cache for exceeding upload and download
limits respectively. Let $s_i$ denote the amount of storage at each cache.
%
%
Table~\ref{tab:notation} summarizes the set of symbols for our problem formulation.

\cut{Another system parameter we consider relates to $u_i$ and $s_i$ at each cache.
The system has some flexibility with the storage space (that can be
utilized for co-operative caching) at each cache, whereas
\cut{For example, a website can incentivize a home to set aside a larger amount
of their local storage to serve content to other users. Similarly,
public hot-spots can be deployed with higher storage capacities.}
the upload
limit of a cache is regulated by the broadband plan at each node, and
determines the amount of data that can be served by a particular cache to
other nodes. In our system, typically, $s_i > u_i$, with the
additional storage quantity to be used for serving repetitive local
requests. Intuitively, it makes sense to set aside a larger storage for
caches which have higher upload limits in order to facilitate co-operative
caching. In this work, we make the design decision
to assign storage proportional to upload limits, i.e. we assume,
$s_i \propto u_i$, and $u_i/s_i = \tau$,
where $\tau(<1)$ is constant across all caches of the system.
}

\begin{table}[ht]
\caption{List of Symbols}
\centering
\begin{tabular}{|c | l|}
\hline
$n$ & \emph{number of caches} \\
\hline
$m$ & \emph{number of objects}\\
\hline
$C_{0\cdots n}$ & \emph{set of caches} \\
\hline
$o_{1\cdots m}$ & \emph{set of objects} \\
\hline
$C_0$	& \emph{content server}\\
\hline
$s_i$ & \emph{storage capacity at Cache} $C_i$\\
\hline
$u_i$ & \emph{upload limit of Cache} $C_i$\\
\hline
$d_i$ & \emph{download limit of Cache} $C_i$\\
\hline
$\alpha_i$ & \emph{price per unit of upload exceeding $u_i$}\\
\hline
$\beta_i$ & \emph{price per unit of download exceeding $d_i$}\\
\hline
$r_{ik}$ & \emph{demand of object $o_k$ at cache $c_i$}\\
\hline
\hline
\end{tabular}
\label{tab:notation}
\end{table}

Let $\{o_1,o_2 \cdots, o_m \}$ be a set of $m$ data objects
that are requested at the caches
over an \emph{epoch}\footnote {We will use the same notation $o_i$
to refer to both the object and its size.}.
In practice, accurate prediction of
the time ordered sequence of requests at every cache may not be feasible. Instead, we assume the
knowledge of an expected/average demand $r_{ik}$ of object $o_k$ at cache $C_i$.
As mentioned earlier, coarse level predictions of this kind is possible by mining
user preferences and mobility patterns.
However, the description of the specific prediction technique is beyond the scope of this
paper.
%
%

At any moment of time, let $\hat{d_i}$ and $\hat{u_i}$ denote the amount of
content that has been downloaded and uploaded respectively by
cache $C_i$, since the beginning of an epoch.
A new request for content $o_k$ at $C_i$ will incur zero cost
if served locally. Otherwise, cache $C_i$ will forward the query
to $C_j$ which has a copy of $o_k$ and fetch it from there.
This operation will increase both $\hat{d_i}$ and $\hat{u_j}$
by the size of $o_k$.
If no cache has $o_k$ available then it will be downloaded from the server.

When all requests are served, any cache $C_i$ incurs
a cost of $\max\{\hat{d_i}-d_i, 0\}\cdot \beta_i$
and $\max\{\hat{u_i}-u_i, 0\}\cdot \alpha_i$
accounting for the extra downloads and uploads, respectively.
The overall cost of content
delivery\footnote{Again, discounting the fixed monthly bill since it appears as an additive
constant.} in the
system is thus $\sum_{i=0}^n \{\max\{\hat{d_i}-d_i, 0\}\cdot \beta_i + \max\{\hat{u_i}-u_i, 0\}\cdot \alpha_i\}$.
We now formally define the data placement problem studied in this paper.

\cut{object $o_k \in O$, has a specified demand $r_{ik}$ at cache $C_i$
There are $m$ data objects $o=\{o_{1\cdots m}\}$ in the system, we will refer to both the object and its size by $o_i$. At each cache $i$, we know the (aggregated) number requests $r_{ik}$ for object $k$.
The cache can serve each of these requests either from its local storage (if it has the object) or by downloading it from one of the other caches. We assume that the server stores all the objects, hence each request can be served at least by getting it from the server.
}
\cut{
At cache $i$, let $\hat{d_i}$ and $\hat{u_i}$ be the amount downloaded and uploaded, respectively, to serve all the requests in the system. If a request for object $k$ at node $j$ is served by downloading from cache $i$ ($i\neq j$), it increases $\hat{d_j}$ and $\hat{u_i}$ by the amount $o_k$. Note that if the object was locally cached ($i=j$), there is no increase in the data transfered.
}
\cut{
When all the requests in the system are served, the node $i$ has to pay a cost of $\max\{\hat{d_i}-d_i, 0\}\cdot \beta_i$ and $\max\{\hat{u_i}-u_i, 0\}\cdot \alpha_i$ for the extra downloads and uploads, respectively. The overall cost of the system is $\sum_{i=0}^n \max\{\hat{d_i}-d_i, 0\}\cdot \beta_i + \max\{\hat{u_i}-u_i, 0\}\cdot \alpha_i$.
}


\medskip
\noindent {\bf Data Placement Problem:} Given the set of caches $\{C_{0...n}\}$ and
demands $R = \cup_{j,k} \{r_{jk}\}$, find the placement of
items $\{o_{1...m}\}$ into caches and a query forwarding policy that satisfies all the
requests in the system with minimum total cost.
\medskip
\vspace{-0.1in}
\subsection{Solution Approach}

We prove that the data placement problem is NP-hard by reducing
it from the \emph{partition problem} $P(O,W,k)$,
defined as follows: Given a set objects $O = \{o_{1\cdots n}\}$ of
weights $W=\{w_{1\cdots n}\}$ ($w_i \in Z^+$)  and a
integer $t<n$, is there a subset $T\subseteq O$ of size $|T| = t$,
such that $\sum_{i\in T} w_i = \sum_{i\notin T} w_i$?
%
We state the following result and refer the proof to Appendix~\ref{app:np-proof}.
%
\begin{lemma}
The data placement problem defined above is NP-hard.
\end{lemma}


We now aim to find a good approximation algorithm to our problem.
%
%
Due to the similarity of our problem with non-metric facility location
problem~\cite{flt-aaim-06}, (which is known to be $O(\log(n))$ hard to
approximate) we believe that our problem should be hard to approximate within
a constant factor. Therefore, from a practical point of view, it would be useful
to find a solution which is within constant factor of the optimum
but stretches the storage and upload limits only by a
constant factor. Such multi-criteria approximation has been
adapted by other researchers~\cite{gm-soda-02} as well. With this motivation,
we design a multi-criteria approximation algorithm for the data placement problem.

\vspace{-.1in}
\section{Related work}\label{sec:related}

Distributed, co-operative content-delivery~\cite{Balke:2005:PDT:1053724.1054042} 
has been very successfully
exploited by peer-peer applications 
(such as BitTorrent\footnote{http://en.wikipedia.org/wiki/BitTorrent}). However,
the model under which they operate is different \-- there is no
centralized control, and any management of the end-node resources is purely
local (e.g. cap outgoing BitTorrent sessions to 32Kbps). In the model
we propose, a central entity collects the end-node information, and while
mindful of the local constraints, makes optimum decision for the entire system.

The notion of co-operative caching has also been studied extensively in the
literature~\cite{kd-tkde-02,gkk-soda-00} before in the 
context of co-operating web proxies, or content
distribution networks (such as Akamai). However, in our work, we assume the
caches to be end-nodes, and the accompanying model
for network bandwidth costs (tiered broadband model) is very different
from the cost model assumed
in previous studies (typically, per byte costs). 
This tiered cost model (particularly relevant
for caching on end-nodes) creates a fundamentally different problem.

%

Baev \etal~\cite{br-soda-01, brs-siam-08} study the 
data placement problem, where the aim is to map items 
and requests to caches to minimize the overall cost. However, unlike us, they consider 
only storage capacities, and not the upload/download capacity at caches.
%
%
The paper closest related to our work is by Guha \etal~\cite{gm-soda-02}, where 
they consider the limit on number of users (or requests) that can be 
mapped to any cache. The crucial difference from our work is the 
assumption that each request served from a cache adds to the constraints or the 
cost (or both), while in our problem the requests served \emph{locally} comes for free.
This changes the problem structure significantly since in our case the optimal 
solution may serve arbitrary number of requests for free (locally), and makes 
it non-trivial to approximate using existing techniques.

Another related body of work is on data placement on parallel 
disks~\cite{gkk-soda-00,kk-ja-06}, where the caches (media servers) have 
storage and load capacities and the aim is to map items and requests 
to caches; however, the aim here is to maximize the number of requests 
served, rather than the cost of serving all requests. Cooperative (proxy) 
caching has also been used for minimizing traffic 
in WWW~\cite{kd-tkde-02,rpr-soda-99}, however, again they do 
not consider upload/download capacities at caches.

Several researchers~\cite{afm-nsdi-05, ird-podc-02} have analyzed the 
performance of peer-to-peer transfers for cooperative file transfers, 
downloading software patches, etc. Like us, they also use 
items (or parts of items) stored at one peer to serve request 
of another. However, since by design these systems are \emph{unmanaged}, 
they mostly focus on either incentive mechanisms (to ensure that peers share) 
or in assisting users to pick ``close by'' peers to download 
from~\cite{krp-imc-05,xyk-sigcomm-08}. In contrast, we have a managed 
system where a central controller predicts requests (based on historical data), 
takes into consideration the limits at caches and makes optimal 
data placement and routing decisions for the entire system.

Several projects~\cite{Ott04,1403006,Hadaller07,Haggle06}
have investigated issues emerging from serving content through 
a network of short-range hot-spots.
Perhaps because the common use-case considered is to provide 
opportunistic Internet access from within moving vehicles, the primary focus 
of such work has been on the performance of the 
wireless link between the user and the hot-spots. 
%
\section{Solution For a Simplified Data Placement Problem}
\label{sec:algos}


The aim of the data placement problem is to minimize the cost of 
satisfying all requests. This cost can be conceptually broken 
into two parts, the cost of putting the desired items in the caches at the
beginning of an epoch
and the cost of serving the requests (using the caches and the
content server) throughout the epoch. We call them \emph{initial placement} and \emph{serving} 
costs, respectively\footnote{The entire initial placement need not 
happen before serving begins, an item may actually be brought 
to the caches during its first request.}. For ease of exposition, we 
start with solving a simplified problem that ignores the initial placement 
cost and optimizes only the serving cost\footnote{Strictly speaking, the 
problem is even more non-trivial, due to the items cached from previous 
epoch (hence requiring cache \emph{replacement}); we will visit this 
issue in Section~\ref{sec:online}.}. The solution of the 
original problem builds upon the steps of this formulation, 
and we will discuss it in Section~\ref{sec:online}. Throughout this 
section, we will discuss the minimization of serving cost only
and refer to it as data placement problem. We will also work with
unit sized objects to keep things simple. However, our
algorithms are easy to modify to account for non-uniform object sizes.

\cut{From a practical point of view, our data placement problem can be thought of as a two step 
process. The first part is to identify the set of items to place in the caches
such that it minimizes the data serving cost. We capture this cost component
in the above formulation. The other part consists of moving
objects from the central server to their respective cache locations(e.g. if
$y_{ik}$ is set to $1$, we need to move the object $o_k$ from the
server to cache $C_i$) - we name it as
the initial placement cost. Ideally, one would like to
minimize the sum over these two cost components. However, for ease of exposition
and as a matter of technical convenience, we ignore the initial placement cost
for the time being and work with the present formulation. We address this 
issue in Section~\ref{sec:online}, where 
we present the complete solution.}

We will now describe our technique to solve the simplified data placement problem and get a multi-criteria approximation.
%
The basic idea is to write the problem as an ILP (integer linear program), solve the relaxed LP to obtain a fractional solution to the data placement problem, followed by a novel rounding approach that ensures a constant factor increase in storage and upload limits that is within a constant factor of the optimum ILP cost.
\cut{
To make things simple, we ignore the cost of initial placement of
items into caches. We address this issue again in Section~\ref{sec:online},
where we discard this assumption and present the complete solution.
}

\vspace{-0.1in}
\subsection{Overall Approach}
\label{sec:overall}
%

\cut{Before going in to the details of the algorithm, we give the intuition behind it. Restating our problem, we have one main server $n$ caches with upload limits $u_i$, download limits $d_i$, storage limits $s_i$. Given set of requests $\{r_{j,k}\}$ at each cache $C_j$, we want to find object placement strategy so that the caches exceeds upload and storage limits by a constant factor and the total cost is with in constant factor of optimal cost of the original problem. The cost here is sum total number of requests going to the main server and the amount of download limit violations at each cache. }

We will now describe the overall structure of our approach.
Our first observation is that if we relax the storage limits slightly, the 
download costs are very easy to handle. For cache $C_i$, we will just use 
an extra storage of $s_i$ (i.e. double the storage of cache $C_i$) 
to store the most frequently 
requested items at $C_i$. 
%
Next, we consider a variant of our problem where any item (or copy of an item, to be
precise) at cache $C_i$
serves at most $\tau_i$ requests, where $\tau_i = \left(\frac{u_i}{s_i}\right)$. 
We show (in Section~\ref{subsec:tau}) that the solution to this has 
identical cost as the original problem.
The advantage of this approach is that if any solution stays 
within storage limits, it is guaranteed to be within 
upload limits as well (or the upload increases by the 
same factor as storage).
%
Rounding this constrained LP turns out to be considerably simpler, 
as it has no upload limit constraints, and so the main non-trivial 
step required is to handle the local requests while moving items 
across caches during rounding (see Section~\ref{subsec:round}).


\cut{Next we will borrow a clever trick used by {\bf refer to guha and munagala}. Let us say that any copy of an object serves at most $\tau$ requests. By doing so, our problem does not change, except that we might be storing extra copies of objects and hence increasing the storage requirements. But, as we see in {\bf algos section}, this does not increase storage by too much. The advantage in doing this is that, as long as we assure that the storage limits are maintained, upload limits will also be maintained. We then just need to worry about storage limits and rest is taken care.

Now, the problem is formulated as an integer linear program and we solve the linear relaxation of it and round the fractional solution to obtain our solution. Let us say $y_{i,k}$ is a variable that denotes if an object $o_k$ is present at cache $C_i$ in the final solution and $x_{i,j,k}$ is the amount of requests for $o_k$ at $C_j$ is going to $C_i$. In our rounding, we first round $y_{i,k}$s to find the object placements and then round $x_{i,j,k}$s to find how to route the requests to appropriate caches. While rounding we will use a very powerful technique introduced by {\bf tardos} called Generalized Assignment Problem or GAP. We will describe the GAP problem in next subsection.}

%
We now mention a result 
due to Shmoys \etal~\cite{st-soda-93} to solve 
the Generalized Assignment Problem (GAP), which can be used to find 
a feasible solution ($z_{ij}$) to the following equations.
%
%
%
\begin{eqnarray*}
\sum_{i=1}^{m} z_{ij} = p_j \ \ \forall j \in [1,n]; \ & \  \sum_{j=1}^{n} z_{ij} \leq q_i  \ \ \forall i \in [1,m]\\
%
%
z_{ij} \in Z^+ \cup \{0\} &\forall i \in [1,m]; j \in [1,n]
\end{eqnarray*}

In~\cite{st-soda-93}, the authors show how to find an integer solution $\{z_{ij}\}$
given any fractional feasible 
solution $\{\hat{z}_{ij}\}$ for the above problem, where the $q_i$
values increases to at most $q_i+1$.
We will refer to the above rounding algorithm as
$GAP(p, q, \hat{z})$, which takes the set of parameters for GAP and
a fraction solution $\{\hat{z}_{ij}\}$ and returns the rounded integral solution $\{z_{ij}\}$.
Observe an interesting point that unlike the original GAP problem, 
we do not have a cost function to minimize, but to
compute any feasible integer solution.

\cut{
With some effort, we can see here that, while
rounding we are essentially keeping the cost same, but violation one of the upload/storage
constraints by a factor of at most $2$. Let us refer the above rounding algorithm as
$GAP(\{n_j\}, \{b_i\}, \{c_{ij}\}, \{x_{ij}\})$, which takes the set of $c_{ij}, n_j, b_i$ and
also a fraction solution $x_{ij}$ to this LP and rounds $x_{ij}$ to integers and returns
the rounded solution, again as $x_{ij}$.}

\vspace{-0.1in}
\subsection{ILP Formulation}
\label{subsec:ilp}

Before describing the ILP, we will make some observations that significantly 
simplify the structure of the ILP.
First, notice that the server ($C_0$) does not download anything as it 
does not have any local requests to serve, and consequently the values 
of $d_0$ and $\beta_0$ do not play any role in our technique.
Further, as mentioned earlier, the server has a high data-rate connection 
that is bought in bulk from the service provider, and hence has typically a 
much cheaper per byte transmission costs as compared to the caches. 
In our formulation, it essentially means 
that $\alpha_0 < \alpha_i, \forall i>0$. As a consequence, 
once a node's upload limit is reached, instead of further using its 
back-haul (for serving another cache), it will be  cheaper to 
source the request directly from the server.
This also intuitively follows from the motivation of this work as we want 
to use only the \emph{spare} upload and storage capacities at nodes to 
reduce the load on the server; however, once no more spare capacity 
is left, the system will fall back to all caches getting content from the server.

%

Based on the above discussion, our caching problem can be formulated 
as the following integer linear program ($I_1$) --

{\footnotesize $$\min \left\{\alpha_0 \cdot \max\{0, \sum_{j,k}x_{0jk} - u_0\} + 
	\sum_{j \in C} \beta_j \cdot \max\{0, \sum_{i,k; i \neq j}x_{ijk} - d_j\}\right\}$$
\vspace{-0.1in}
\begin{eqnarray}
s.t.
\sum_{j \in C,k \in O; i \neq j}  x_{ijk} \leq u_i &\forall i\in C \\
\sum_{k \in O}  y_{ik} \leq s_i &\forall i\in C \\
\sum_{i\geq 0} x_{ijk} = r_{jk} &\forall j\in C; k \in O\\
0 \leq x_{ijk} \leq r_{jk} y_{ik} &\forall i,j\in C; k \in O\\
%
%
%
y_{ik} \in \{0,1\}, \ \ x_{ijk} \in {Z}^+ \cup \{0\} &\forall i,j\in C; k \in O
%
%
\end{eqnarray}
}

Here $y_{ik}$ denotes that the object $k$ is stored at cache $i$, 
and $x_{ijk}$ is the number of requests for object $k$ at 
cache $j$ that are served from cache $i$. The first two constraints 
express upload and storage limits. The third constraints specifies 
that all the requests must be satisfied, while the fourth makes sure 
that data is only served from a cache that stores it. Here $C$ denotes
the set of caches excluding the content server and $O$ is the set of 
objects.

The cost function has two components, the first is the cost
of serving all the requests from $C_0$ i.e. the content server. The second term 
captures the penalty incurred by the caches
(excluding the content server) for exceeding the download limit, if any.
As mentioned earlier, since $\alpha_0 < \alpha_i$, once the upload limit
is reached for a cache, the extra requests are served by the content server.
Therefore, the upload limits are satisfied as hard constraints and there is
no penalty term associated with the uploads.
It is easy to see that the non-linear terms in the cost function 
can be converted to linear constraints, using standard techniques. 

Note that the download cost must be accounted in the cost function, as it 
may not be possible to get a feasible solution where all requests 
are satisfied without crossing the download limits; 
however, our assumption that \bytecost\ of server is lower 
than those of caches implies that no requests are served by a 
cache after reaching its upload limit.
We now consider both the terms in the cost function in turn. First we
get rid of the second term through the following observation.

\medskip

\noindent{\bf Handling Download Costs:} We first show how to take care 
of the download limit constraints. For cache $C_j$, 
let $O_j\subseteq O$ be the set of $s_j$ items that has 
the maximum total request ($\sum_{k\in O_j} r_{jk}$) at $j$. 
Clearly, storing these items in $C_j$ minimizes the amount of requests going out, 
and is at most the requests going out of $j$ in the optimal solution of $I_1$.
%
%
%
So we increase the storage at $j$ by $s_j$ (i.e. double the storage at $C_j$) 
and use the extra storage to store the item set $O_j$. 
This ensures the second term in the cost is at most the second 
term in the optimal solution to $I_1$\footnote{Note that we 
do not really need to store $s_j$ items; we can compute 
the optimal solution to $F_1$ and store only $\sum_k y_{jk}$ items at $j$.}.
For the rest of this paper we consider only the first term in the cost function of $I_1$.
After this simplification, our cost function looks like 
$$\left\{\alpha_0 \cdot \max\{0, \sum_{j,k}x_{0jk} - u_0\}\right\}$$

We make a further observation that if one attempts to minimize the number of 
requests received at the content server, then it automatically minimizes the cost of the
first term in $I_1$, and {\it vice versa}. The cost is zero if it is less than $u_0$,
else positive. Ignoring the multiplicative factor $\alpha_0$, our cost function now
has a simple form $\sum_{j,k}x_{0jk}$. In the next subsection we present an
equivalent optimization problem which attempts to minimize this cost function.

\cut{
Consider a solution of $F_1$, and let $o_j$ denote the (fractional) number of items stored at cache $i$ ($o_j = \sum_k y_{ik}$).
The total number of requests going out of $j$ is $R_j - \sum_{i,k; i = j} x_{ijk}\cdot r_{jk}$, where $R_j$ is the total requests on cache $j$ and the second term is the amount served locally due to items served at cache $j$.
Let us consider the item set $O_j$ containing $\lceil o_j \rceil$ items with maximum number of requests ($r_{jk}$ values) at cache $j$. Clearly, the number of requests satisfied locally is at least as much as those in the optimum fractional solution, i.e. $\sum_{k\in O_j} r_{jk} \geq \sum_{i,k; i = j} x_{ijk}\cdot r_{jk}$.
So if we store the best $\lceil o_j \rceil$ items at cache $j$, the amount of requests going out (and hence the value of $\beta_j$) would be at most as much as the solution to $F_1$, while increasing the storage by at most an additive factor of $\lceil o_j \rceil \leq s_j$.

So to take care of download limits, we solve the $F1$, find the sets $O_j$ for each node and store them at node $j$.

This ensures that the download costs paid by $j$ is lesser than that in the optimal (fraction) solution, and needs only $s_j$ extra storage at $j$. For the rest of this paper we consider only the first term in the cost function of $L_1$.
}
\subsection{$\tau$-Constrained Placement}
\label{subsec:tau}

Let us consider a variation of our problem where any item stored at a 
cache $C_j$ serves at most $\tau_j$ number of requests\footnote{Recall that
$\tau_j = \left(\frac{u_j}{s_j}\right)$}. Since the storage at 
cache $C_j$ is $1/\tau_j$ times the upload limit, this immediately implies 
that if any placement stays within the storage limit for a cache, it 
also stays within the upload limit. 
%
The resulting $\tau$-constrained ($\tau$-LP) $I_2$ is as follows.
%
%

{\footnotesize $$\min \sum_{j,k}  x_{0jk} $$
\begin{eqnarray}
s.t.
\sum_{j; i \neq j} x_{ijk} &\leq  \tau_i \cdot y_{ik} &\forall i\in C; k\in O \\
\sum_{k} y_{ik} &\leq  s_i &\forall i\in C \\
\sum_i x_{ijk} & = r_{jk} &\forall j\in C; k \in O\\
0 \leq x_{ijk} &\leq  r_{jk} y_{ik} &\forall i,j\in C; k \in O\\
x_{ijk}, y_{ik} &\in  Z^+ \cup \{0\} &\forall i,j\in C; k \in O
\end{eqnarray}
}
Here the first constraint specifies that each copy of an 
item can serve only $\tau_i$ requests, while the last 
constraint gives freedom to open multiple copies of the same item at a cache.
It is important to note that the restriction of $\tau_i$ is only on the 
requests served to other caches (constraint 6); a single copy of an 
item can serve any number of \emph{local} requests, but 
not more than $\tau_i$ \emph{external} 
requests.


The following lemma shows that a (approximate) solution to $I_2$ 
gives a (approximate) solution to $I_1$ with the same cost. 
%

\begin{lemma} \label{LEM:REDUCTION}
For any $\tau$ s.t. $1/\tau \in Z^+$, any solution of $I_2$ 
with $\alpha$ times the optimum cost can be converted to a solution 
of $I_1$ also of $\alpha$ times the optimum cost, by 
increasing the storage at caches by at most a factor of 2.
\end{lemma}

\begin{proof}
Let $y_{ik}$ be any (even optimal) solution to $I_1$. For any item $o_k$ placed
at cache $C_i$ (i.e. $y_{ik}$=1), let $\ell_{ik}$ denote the number of requests of
other caches satisfied by it. We create the solution to the $\tau$-constrained
placement by opening $\lceil \ell_{ik}/\tau_i \rceil$ copies of $o_k$ at cache $C_i$,
and dividing the $\ell_{ik}$ requests among these copies. Clearly, each copy serves
at most $\tau_i$ requests. Moreover, at storage at cache $C_i$ is
$\sum_k \lceil \ell_{ik}/\tau_i \rceil \leq 2 \sum_k \ell_{ik}/\tau_i \leq 2 u_{i}/\tau_i \leq 2 s_i$.
To convert the solution of $\tau$-LP to $I_1$, we simply set all
$y_{ik}>1$ to $y_{ik}=1$. Observe that this step only decreases the storage
requirement. The total upload required for Cache $C_i$ is 
$\sum_k \sum_{j;i\neq j} x_{ijk} \leq \sum_k \tau_i\cdot y_{ik} \leq \tau_i\cdot s_i=u_i$, 
Therefore, it is a feasible solution for $I_1$ with the same cost as that of $\tau$-LP.

By combining these two facts, the result of the lemma follows.
\end{proof}
%

\subsection{Rounding the $\tau$-constrained LP}
\label{subsec:round}

\begin{algorithm}[tbp]
\caption{\round($C=C_{1\cdots n}, R=\{r_{ik}\}, \hat{y}, \hat{x}$)}

\label{algo:round}

\begin{algorithmic}[1]
\FORALL{$i,k$ s.t. $\hat{y}_{ik} \geq 1/2$}
\STATE $y_{ik} = \lceil \hat{y}_{ik} \rceil$;
\ENDFOR
\FORALL{$k$ s.t. $\exists i: \hat{y}_{ik} < 1/2$}
\STATE $n_k \leftarrow \displaystyle{\sum_{\substack{i: \hat{y}_{ik} < 1/2}} \hat{y}_{ik}}$;
\STATE {\bf if}($n_k < 1/2$) {\bf then} reroute all demands coming to $n_k$ to 
any open copy of $o_k$;
\STATE {\bf else} scale all $\hat{y}_{ik} < 1/2$ to $\hat{y}_{ik}\cdot \lceil n_k \rceil/n_k$;
\ENDFOR
\STATE Let $s'_i$ be the amount of storage occupied by items $o_k$ such that $\hat{y}_{ik} < 1/2$;
\STATE $y \leftarrow GAP(n_k, s'_i, \hat{y})$;
\STATE reroute all demands coming to $n_k$ to the copies opened by GAP;
\FORALL{$i,j,k$ s.t. $\hat{x}_{iik} < \hat{y}_{ik}$}
\STATE $\hat{x}_{ijk} \leftarrow \hat{x}_{ijk} / (r_{ik}-\hat{x}_{iik})$; \hfill /* handle extra demands */
\ENDFOR
\STATE {\bf for all} {$k$} {\bf do} \hfill /* round $x$ values */
\STATE \ \ \ $x_{\{\cdot\}\{\cdot\}k} \leftarrow$ GAP($\tau\cdot y_{\{1\cdots n\} k}, r_{\{1\cdots n\}k}, \hat{x}_{\{1\cdots n\}\{1\cdots n\}k}$);
\STATE {\bf end for}
\end{algorithmic}
\end{algorithm}

Let $F_2$ be the fractional version of the ILP $I_2$, 
that relaxes the last constraints to $y_{ik}, x_{ijk}i\geq 0$.
For subsequent discussions, lets assume that $\{\hat{y}_{ik}, \hat{x}_{ijk}\}$ denotes
the optimum solution of the LP. The rounding procedure will convert this fractional
solution to an integral solution $\{y_{ik},x_{ijk}\}$
Rounding $\hat{y}_{ik}$ values in the fractional solution should be 
straightforward, as it is just an instance of the GAP problem.
\cut{
Rounding the fractional solution requires two stages, first to round the item placements ($y_{ik}$) and then to round the demand routing ($x_{ijk}$). In the first stage, we already have fractional open copies, all we need to do is to open at least those number of integral copies for each item (so that they can handle all requests coming to them) and ensure that the storage is not exceeded by much. This is an instance of the general assignment problem (GAP), can be solved by techniques provided by Shmoys \etal\cite{st-soda-93} (as discussed below).}
However, in our problem there is an extra complication---the 
requests satisfied locally ($\hat{x}_{iik}$) do not add to the cost or the 
upload capacities, but due to reassignment of items to caches 
these will also start contributing to the cost and capacities.
To avoid such extra requests, we first make sure that the items 
serving large values of $\hat{x}_{iik}$ are not moved to other caches. 
Specifically, if we round \emph{big enough} copies of 
any item ($\hat{y}_{ik} \geq 1/2$) to the closest higher integer, 
the extra requests are also bounded by $\hat{x}_{iik} < r_{ik}/2$, 
due to constraint (9).

We now describe the steps in our procedure to 
round the fractional solution of $F_2$ (see Algorithm \ref{algo:round}).

\medskip

\noindent {\bf (Step 1)} For any $\hat{y}_{ik}\geq 1/2$, 
round it to $\lceil \hat{y}_{ik} \rceil$.

\noindent {\bf (Step 2)} 
Let $\hat{n}_k = \displaystyle{\sum_{\substack{i: \hat{y}_{ik} < 1/2}} \hat{y}_{ik}}$
and $n_k = \lceil \hat{n}_k \rceil$. If $\hat{n}_k < 1/2$, 
since $\sum_i \hat{y}_{ik} \geq 1$ (due to constraints 3 and 4), 
there is at least one copy of $k$ opened in the last step. 
Route the demands coming to the copies of item 
contributing to $\hat{n}_k$ to the integral open copy.

\noindent {\bf (Step 3)} For $\hat{n}_k \geq 1/2$, 
scale each $\hat{y}_{ik}$ by a factor of $n_k/\hat{n}_k$, and 
hence with the new values, 
$n_k =\displaystyle{\sum_{\substack{i: \hat{y}_{ik} < 1/2}} \hat{y}_{ik}}$.

\noindent {\bf (Step 4)} Let $s'_i$ be the amount of storage occupied by the 
items in step (3) at cache $C_i$. Find the integral placement of all 
objects in step (3), by solving the GAP problem: GAP($n, s', \hat{y}_{ik}$). 
Note that the scaling in last step ensures that the $\hat{y}_{ik}$ values 
for each item sum up to an integral value ($n_k$), which 
is required by the GAP problem.

\noindent {\bf (Step 5)} Now assign all the remaining 
demands ($\hat{x}_{iik}$) proportionally to the caches that 
are already serving the rest of demands for $r_{ik}$.

\noindent {\bf (Step 6)} To round $\hat{x}_{ijk}$ values, note that 
only constraints (6) and (8) are important, since 
values of $\hat{x}_{ijk}$ do not affect (7), and (9) follows from (3).
For each item $k$, this essentially means 
solving the GAP($\tau\cdot y_{\{1\cdots n\} k}, 
r_{\{1\cdots n\}k}, \hat{x}_{\{1\cdots n\}\{1\cdots n\}k}$) problem 
(note that $k$ is fixed in the subscripts of this GAP instance).

\medskip

\cut{
We now want to find an integral placement of $\cup_k n_k$ objects on the caches of sizes $s'_i$. This requires us to find an integer feasible solution to the following LP.
\begin{eqnarray*}
\sum_{i} y_{ik} &= n_k & \forall k\in O \\
\sum_{k}  y_{ik} &= s'_i & \forall i\in C\\
y_{ik} &\in  Z^+\cup \{0\} &\forall i\in C; k \in O
\end{eqnarray*}

Observe that the $\hat{y}_{ik}$ values are a feasible (fractional) solution to this LP.
Further, this is an instance of the general assignment problem (GAP) problem, and due to~\cite{st-soda-93}, if there is a feasible fractional solution a integral solution where $s'_i$ increase by a factor of $2$ that can be computed in polynomial time. Observe an interesting point unlike the GAP problem, we do not have a cost function to minimize, but rather just want to compute any feasible integer solution.

The above steps give an integral solution that serves all the demands going across the caches in the fractional solution. However, due to the object movement, the requests satisfied locally ($x_{iik}$ at cache $i$ for object $k$) may now need to go to some other cache. Notice that since $x_{iik}\leq r_{ik}y_{ik}$, and we have not moved any object with $y_{ik}>1/2$, $x_{iik}<r_{ik}/2$. In other words, at least half of the demands of $r_{ik}$ is already satisfied.}

\begin{algorithm}[tbp]
\caption{\caching($C, R$)}

\label{algo:caching}

\begin{algorithmic}[1]
\STATE {\bf for all} {$i \in C$} {\bf do} \hfill /* handle download limit */
\STATE \ \ \ put $s_i$ items with max $r_{ik}$ values at $i$;
\STATE \ \ \ $s_i \leftarrow 2\cdot s_i$;  \hfill /* for solving the $\tau$-LP */
\STATE {\bf end for}
\STATE $\langle x,y \rangle:=$ fractional solution to the $\tau$-LP;
\STATE $\langle x,y \rangle:=$ \round$(C,R,x,y)$;
\STATE {\bf for all} {$i,k$} {\bf do} \hfill /* convert to solution of $I_1$ */
\STATE \ \ \ {\bf if}($y_{ik} > 1$) {\bf then} $y_{ik} = 1$;
\STATE {\bf end for}

\end{algorithmic}

\end{algorithm}


The following lemma states the approximation ratio of our rounding procedure. 
%
\begin{lemma}
\label{LEM:ROUNDING}
The \round\ algorithm finds a $\tau$-constrained placement with 
at most $8$ times the optimal cost, and requires $(4s_i + 2)$  
storage and $(8u_i + 4 \tau_{max})$ upload at cache $C_i$, where $\tau_{max} = \max_i \tau_i$.
\end{lemma}

\begin{proof}
Let us first look at the upload limits. 
Let $u^1_i$, $u^2_i$ and $u^3_i$ be the upload used by items 
processed in steps (1) (2) and (3). Note that since 
they are disjoint, the total upload is $(u^1_i+u^2_i+u^3_i)\leq u_i$.
The demands coming to copies of $y_{ik}$ closed in step (2) 
get rerouted to (one of) the copies opened in step (1). 
Further, since $n_k<1/2$ and it is routed to a $y_{ik}>1/2$, only 
the items contributing to $u^1_i$ get at most 
twice the number of requests. Hence this step increases the total upload from 
cache $i$ by an additional amount $u^1_i$.
The scaling of $n_k$ values in step (3) increases the $s'_i$ 
values passed to the GAP problem by a factor of 2, accompanied by 
another additive factor of 1 (due to GAP solution, 
see Section~\ref{sec:overall}). Since the solution of GAP 
can use the entire $s'_i+1$ storage, the upload by items contributing 
to $u^3_i$ can now grow to $2u^3_i + \tau_{max}$.
Hence at the end of step (4), the upload has increased to 
$2(u^1_i + u^3_i) + \tau_{max} \leq 2u_i + \tau_{max}$.

Notice that in step (5) since we only assign requests 
$x_{iik}\leq r_{ik}y_{ik}$, and we have not moved any object 
with $y_{ik}>1/2$, $x_{iik}<r_{ik}/2$. In other words, at least 
half of the demands of $r_{ik}$ is already satisfied. 
Hence this step can at most double the amount of upload from any cache.
Further, rounding $x$ values in last step gives 
another at most $2$ factor increase in the total upload, 
due to the GAP solution.
Combining these, we get the final result.

Following a very similar argument for storage, we get the 
final storage at cache $i$ to be $4s_i + 2$.
The increase in cost (basically the upload by the server), 
follows the same pattern as the uploads on other caches. Specifically, 
it is affected only in steps (2), (5) and (6), doubling 
each time. Hence the final cost is at most $8$ times the optimum.
\end{proof}

%


%

The above lemma along with the extra $s_i$ storage per cache to 
handle the download limit gives us the following theorem.
\begin{thm}
\label{thm:apx}
The \caching\ algorithm gives a solution with at most $8$ 
times the optimal cost, and requires $(5s_i + 2)$ storage 
and $(8u_i + 4 \tau_{max})$ upload at cache $C_i$.

We consider $\tau_{max}$ as a system parameter that can be tuned as per the requirements.
If $\tau_{max}$ is set to a constant, it will require storage at the caches $s_i$ to be 
proportional to the upload limits ($u_i$). This is not a concern since in general, 
per byte cost
for storage is much lesser than the network cost.
\end{thm}

\section {Solution for the Data Placement Problem}
\label{sec:online}

The solution presented in the previous section needs two additions to solve 
the data placement problem defined in Section~\ref{sec:model}. 
First, we need to include the initial placement cost in the objective functions of the 
LP. Further, as our aim is to determine the item placement at 
each epoch, the solution should take into account the items already in the caches 
due to placements in the previous epoch; this essentially implies designing a 
cache \emph{replacement} strategy that determines which new item replaces which 
old item for the next epoch.
In this section, we extend our rounding procedure to both these scenarios.

\subsection{Including the Initial Placement Cost}
\label{subsec:initcost}

We first need to add the initial placement cost to the 
objective function in $\tau$-LP ($I_2$).
Notice that in the solution of $\tau$-LP, there could be 
multiple copies of an object at a cache (as $y_{ik}$ could be  greater than $1$), so 
we could not directly use $y_{ik}$ to calculate the initial placement cost. 
Ideally, we need to change it to $\min \sum_{j,k}  x_{0jk} + \sum_{i,k} z_{ik}$, 
where $z_{ik}$ is 1 if a copy of $o_k$ is stored in cache $C_i$, 0 otherwise.
The problem with this formulation is that it is not possible to 
express the $z_{ik}$ constraints in terms of linear inequalities, 
which makes linear programming approach not suitable for this. 

\cut{
}
We tackle this problem by using a modified function that approximates
$z_{ik}$ within a constant factor.
We use $\frac{y_{ik}}{y_{ik}+1}$ as a $2$-approximation 
to $z_{ik}$; as $z_{ik}/2 \leq \frac{y_{ik}}{y_{ik}+1}\leq z_{ik}$. We 
change the cost function in $I_2$ (while keeping all constraints same) to:
$$\left\{ \sum_{j,k}  x_{0jk} + \sum_{i,k} \frac{y_{ik}}{y_{ik}+1}\right \}$$ 
Let us call this new integer program as $I_3$. 
Clearly this function is $2$-approximation to cost function in 
$I_2$ for all $x$ and hence any $\beta$-approximation to $I_3$ 
should give us $2\beta$-approximation to $I_2$.
However, the relaxation of $I_3$ from integer values to real 
values for $x_{ijk}$'s and $y_{ik}$'s is unfortunately a concave optimization
problem for which no general technique is known to solve at optimality.
Therefore, we use standard branch and bound technique to get a probably good fractional
solution and use a rounding technique to convert it to integral values.


\eat
{
Clearly, the second term does not give the exact placement cost, but it is a good approximation to the cost; it is 0 when $y_{ik}=0$, $1/2$ at $y_{ik}=0$ and approaches 1 as $y_{ik}\rightarrow \infty$. More generally, it easy to show that $$\frac{z_{ik}}{2} \leq \frac{y_{ik}}{y_{ik}+1} \leq z_{ik}$$
We can also prove that any values of $y_{ik}$ that minimizes the above equation also approximates the term $\sum_{i,k} z_{ik}$ to at most 2 factor of the optimum; hence optimum solution of $I_3$ approximates the integer program with initial placement cost within a factor of two.
Furthermore, an $\alpha$-approximate solution to $I_3$ will yield a $2\alpha$ approximation for the original problem. \comm{Do we need to add proofs of these to appendix? --Nisheeth}

Now, let $F_3$ denote the fractional version of $I_3$ that relaxes $x_{ijk},y_{ik} \geq 0$. Note that the cost function in $F_3$ is convex over a set of linear constraints and therefore, is known to admit a polynomial time $(1+\epsilon)$-approximate solution for any given $\epsilon > 0$~\cite{conv-opt}.\comm{Naidu, add this citation. --Nisheeth}
}
Starting with a fractional (approximate) solution of $F_3$, we follow the very similar rounding steps outlined in Section~\ref{subsec:round} and convert it to an integral solution of $I_3$.
We summarize the results in the following theorem and refer the details of
the proof to Appendix~\ref{app:online-proof}.
%

\begin{thm}
\label{THM:ONLINE}
The Branch and bound technique on $F_3$ followed by \round\ algorithm finds a solution to 
$I_3$ with cost at most $16$
times the fractional cost and requires $(5s_i+2)$ units of storage 
and $(8u_i+4\tau_{max})$ units of upload limit at cache $C_i$.
\end{thm}

\cut{\begin{proof}
Owing to the page limits, the proof is deferred to the full version
of the paper~\cite{mango-itds}.
\end{proof}}

\subsection{Including Cache Replacements}
\label{sec:epoch}

In our model, at the beginning of every epoch, we have to recompute the data placement based on the predicted values of requests at each cache (Section~\ref{sec:model}).
Obviously, the cache will have many items already present due to placement of previous solution. We now modify our solution to take into account this previous state as well.
Let $Y_i^{t}$ denote the set of items that are present in cache $C_i$ in epoch $t$. Clearly, the initial placement cost of putting any item in $Y_i^{t-1}$ in cache $C_i$ for epoch $t$ is zero.
Therefore, in epoch $t$, we solve the data placement problem with the slightly
modified objective function.
$$\left\{\sum_{j,k}  x_{0jk} + \sum_i \sum_{k\notin Y_i^{t-1}} \frac{y_{ik}}{y_{ik}+1}\right\}$$

\cut{
Observe that this cost function is also convex, as the second term is also just a sum over the (a subset of) terms $y_{ik}/(y_{ik}+1)$ (that are chosen based on the previous cache state).
Hence, similar to previous section, the $(1+\epsilon)$-approximation~\cite{conv-opt} followed by our rounding procedure gives a multi-criteria approximation to the data placement problem with replacements.
We state the final result of our paper in the following 
theorem (the proof is very similar to that of Theorem~\ref{THM:ONLINE}).
}
We again use branch and bound technique to get an approximate solution to the above program
and use the \round\ algorithm to convert that solution to integer values.
\cut{
\begin{thm}
The convex optimization~\cite{conv-opt} followed by \round\ algorithm finds a solution to the data placement problem with cost at most $16(1+\epsilon)$
times the optimum cost and requires $(5s_i+2)$ units of storage 
and $(8u_i+4\tau)$ units of upload limit at cache $C_i$.
\end{thm}
}

\begin{thm}
The branch-and-bound technique followed by \round\ algorithm 
finds a solution to the data placement problem with cost at most $16$
times the fractional cost and requires $(5s_i+2)$ units of storage 
and $(8u_i+4\tau_{max})$ units of upload limit at cache $C_i$.
\end{thm}
The above result states that we can get 
probably good
solution for the data placement problem in every epoch. It does not, 
however, give guarantees on the solution over multiple (or all) epochs. 
For example, a scheme that has knowledge of requests across multiple future epochs 
may get items that would reduce the cost (specifically, the initial placement cost) 
not just for the current epoch, but for subsequent future epochs as well.
We believe such a scheme would not benefit too much as compared 
to our scheme because of two reasons. First, it would require one to
predict requests that are far into the future, which would be both 
impractical and inaccurate, limiting the usefulness of this scheme. 
Second, 
in our experiments (Section~\ref{sec:simulations}), we observed that 
the initial cost is a very small fraction ($\le 15\%$) of the total cost; 
hence the gain from optimizing it is not expected to result in significant cost reduction.
%
%

\cut{In our experiments (Section~\ref{sec:simulations}), we observed that 
typically 5-10\% of the upload limits are spared
in an epoch. We use these spare upload limits in the following way to 
bring down the initial placement cost further. If some object is to be placed on multiple
caches, then downloading it each time from the server incurs unnecessarily high cost.
Instead, we download a single copy of the object from the server, and then 
utilize the spare upload limits at the caches (from the previous epoch) to re-route
the objects to its destination.
}

\medskip 

\noindent{\bf A Lower Bound:} We now provide a simple lower bound on the total cost of the data placement problem. For cache $C_i$, define a function $f_{ik} = r_{ik} -1$ if $k$ was present in cache in previous epoch and $f_{ik} = r_{ik}$ otherwise, that defines the amount of transmission saved by caching $o_k$ for the next epoch. Now, suppose we cache at $C_i$ the set $F_i$ (of size $|F_i| = s_i$) of items with the maximum $f_{ik}$ values.
Clearly, this technique completely ignores cooperation among caches and minimizes (optimally) only the cost of serving local demands on each cache; we call it the \emph{local caching} algorithm and compare our solution with it.
The cost of this placement is $\sum_{i, k}{r_{ik}} - \sum_{i, k\in F_i} f_{ik}$, and it is easy to show that this is at most $\sum_i u_i$ away from the optimum. We compare our algorithms with this lower bound on the cost of the optimum solution (that we will call $C_{lb}$), and our results show that our approach very close to it (always with in $20\%$, in many cases with in $3\%$ of $C_{lb}$).

\section{Optimization}
\label{sec:ext}

Computation of a low-cost data placement strategy needs to 
solve the LP $I_2$ optimally. However, as the number of caches 
increases, the running time of the LP solver becomes prohibitively large
(a sample run of 200 caches takes more than a day using the 
CPLEX solver). This increased running time is intuitive, owing
to the large complexity of the LP solvers. 
%

We propose a cache clustering technique to speed up 
the computation of the fractional LP solution.
The idea is to group nodes into clusters, solve the fractional $\tau$-constrained 
LP for each cluster separately, and then apply the rounding 
procedure on the combined solution of the clusters. The key 
benefit is that running LP on small clusters is very fast; 
however, we risk losing benefits of cooperation among caches 
that were put in different clusters.
Our clustering procedure 
attempts to maximize 
the intra-cluster cooperation by grouping caches based on the similarity 
of the items they request.
The intuition is that if caches with similar requests 
are put into the same cluster, they will benefit more from 
cooperating with other caches in the cluster (in serving 
the common requests), than with caches in other 
clusters with dissimilar requests. 


We use a simple hierarchical bottom-up clustering 
to compute the clusters.
At each iteration, we greedily pick the pair of clusters 
which has the maximum similarity and merge them.
We define similarity sim($i,j$) between clusters $F_i$ 
and $F_j$ as the \emph{Jaccard's coefficient} of their 
request sets $R_i$ and $R_j$, respectively, 
sim($i,j)=\frac{|R_i \cap R_j|}{|R_i \cup R_j|}$.
Our technique takes only a single parameter $k$, which denotes 
the maximum size of any cluster. We merge clusters only if 
combining them produces a cluster of size at most $k$.
Note that we do not merge clusters if the request sets 
do not overlap; this is because if the caches in the two clusters do 
not have common requests, they will not benefit from cooperative caching.
{
\begin{algorithm}[t]
\caption{\clust$(C, k)$}
\label{algo:clust}
\begin{algorithmic}[1]
\medskip
\STATE $F_{1...n}=C_{1...n}$;
\WHILE {$|F| > 1$}
\STATE choose a pair of clusters $F_i$, $F_j \in F$ s.t. $|F_i\cup F_j| \leq k$ with maximum
	value of $\frac{|R_i \cap R_j|}{|R_i \cup R_j|};$
\STATE \textbf{if} ($|R_i \cap R_j| \leq 0$) \ \  \textbf{break};
\STATE $F \leftarrow F - \{F_i, F_j\}\cup \{F_i \cup F_j\}$; /* merge clusters */
\ENDWHILE
\RETURN $F$;
\end{algorithmic}
\end{algorithm}
}

{
\begin{figure*}[tbp]
\begin{minipage}[p]{0.3\textwidth}
\centering
\includegraphics[width=\textwidth]{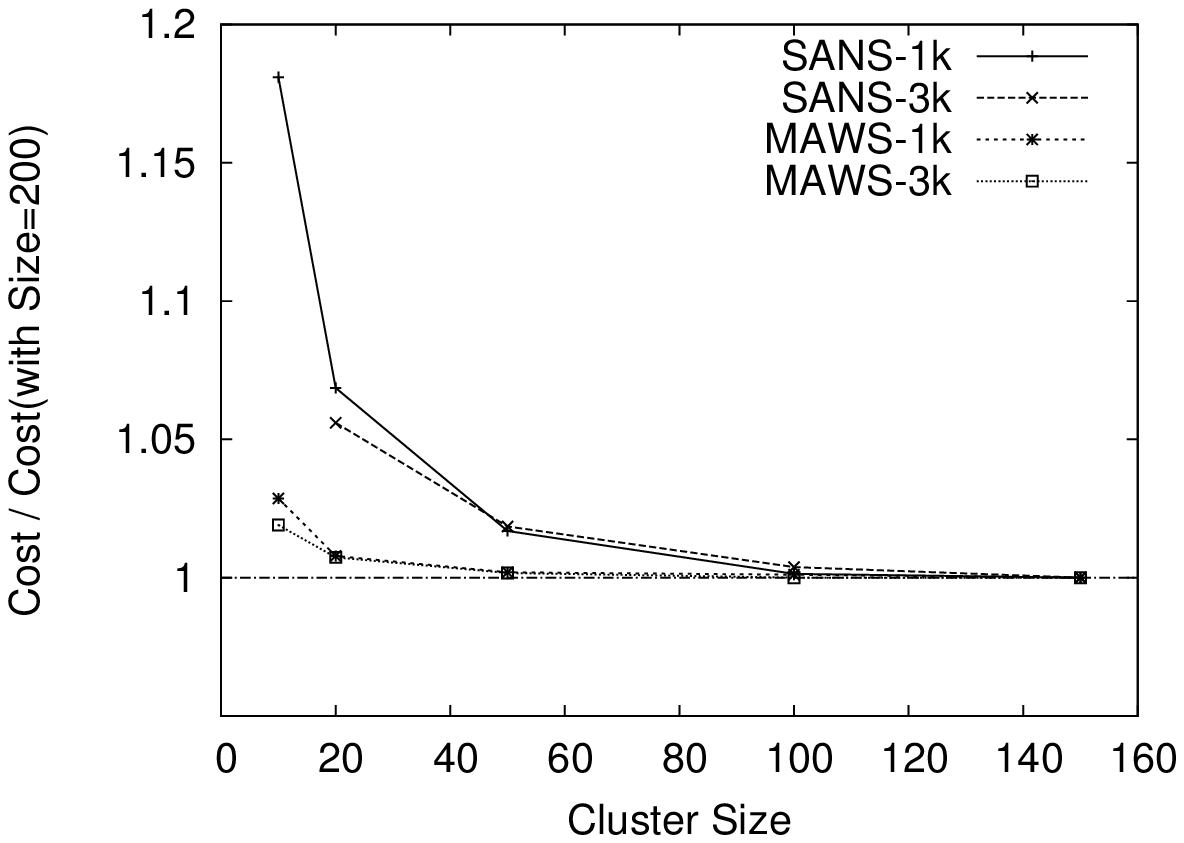}
\vspace{-.2in}
\caption{Costs for different cluster sizes}
\label{fig:cost-vs-clustsize}
\end{minipage}
\begin{minipage}[p]{0.3\textwidth}
\centering
\includegraphics[width=\textwidth]{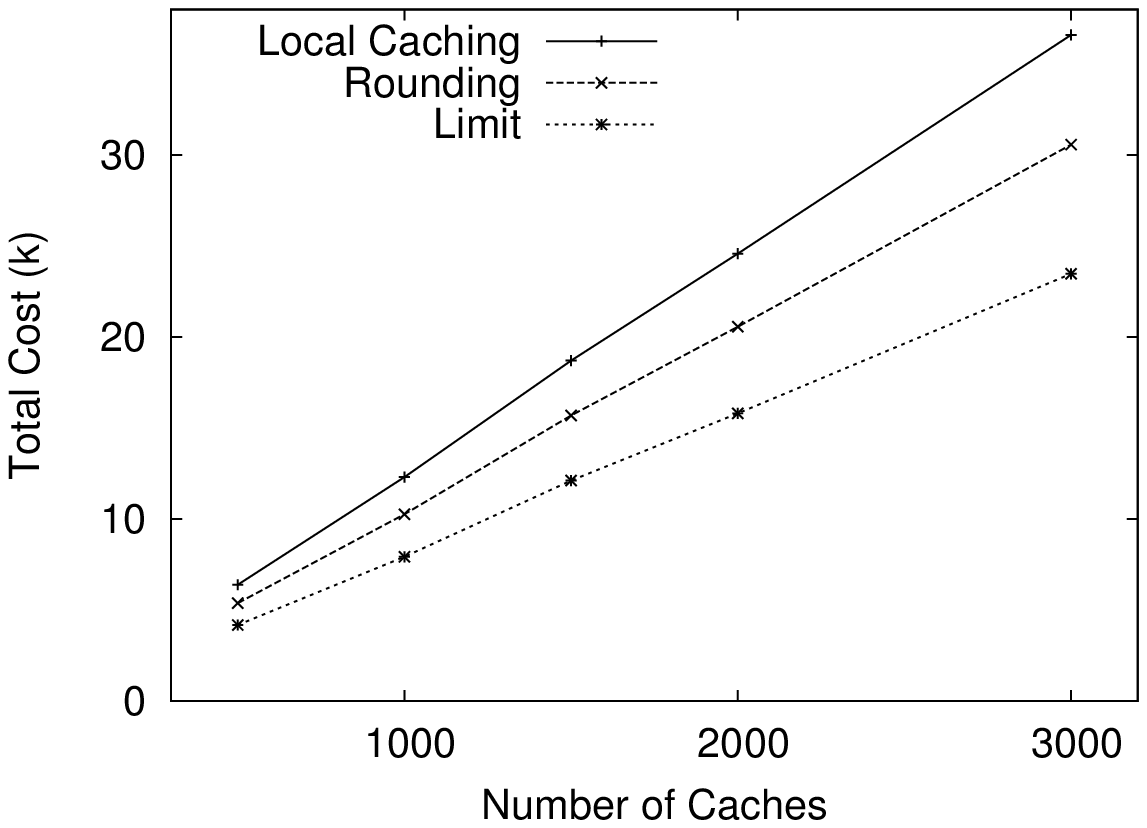}
\vspace{-.2in}
\caption{$C_{round}$ vs. $n$(MAWS).}
\label{fig:cost-vs-ncaches-mango}
\end{minipage}
\begin{minipage}[p]{0.3\textwidth}
\centering
\includegraphics[width=\textwidth]{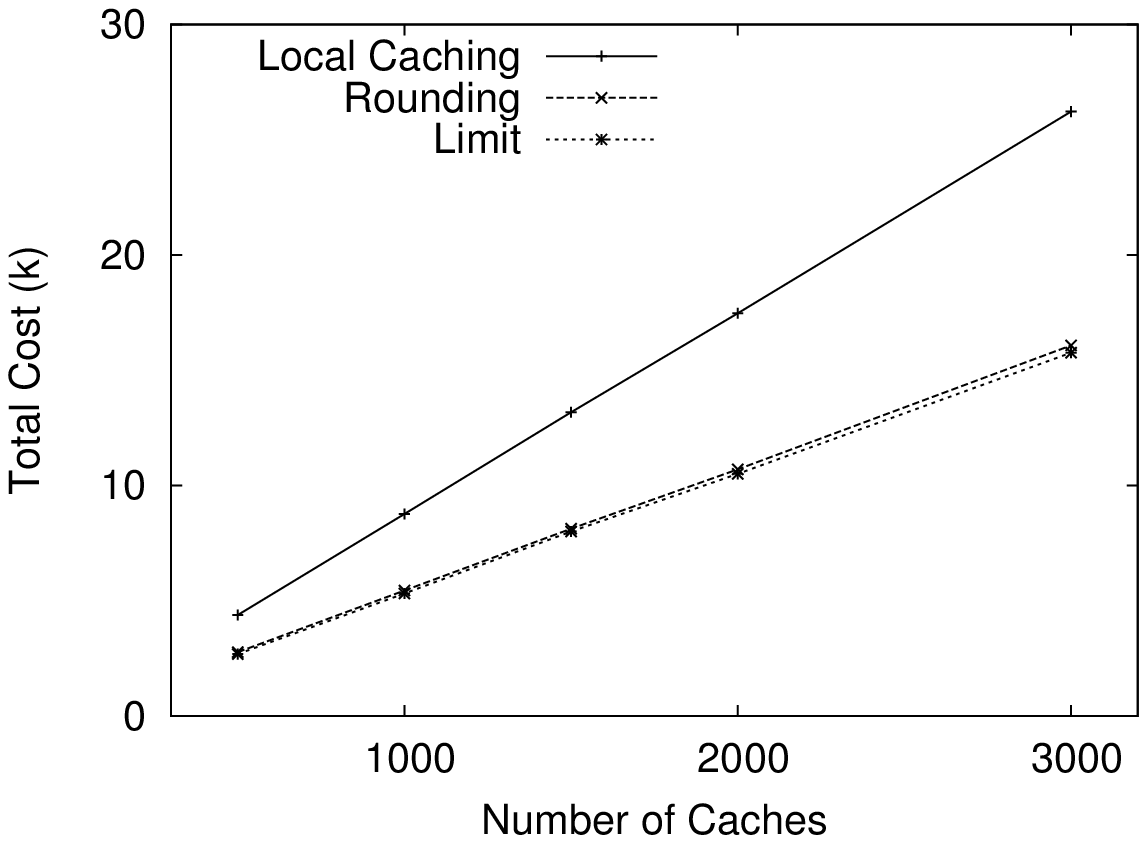}
\vspace{-.2in}
\caption{$C_{round}$ vs. $n$ (SANS).}
\label{fig:cost-vs-ncaches-netflix}
\end{minipage}
\end{figure*}

\cut{
\begin{figure*}[htpb]
\centering
\includegraphics[width=\textwidth]{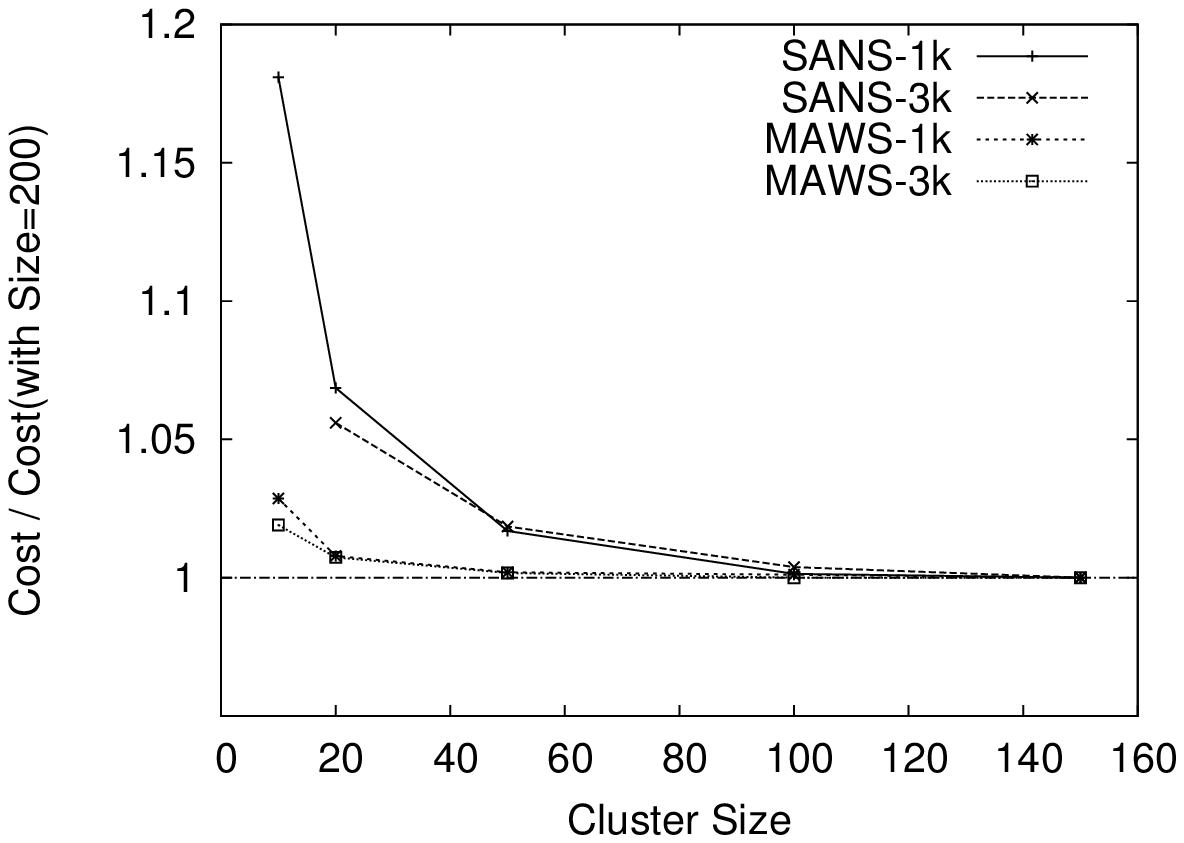}
\vspace{-.2in}
\caption{Costs for different cluster sizes}
\vspace{-.2in}
\label{fig:cost-vs-clustsize}
\end{figure*}
}
\section{Simulations} \label{sec:simulations}

We now describe an extensive set of simulations to evaluate our techniques.
First, we describe our experimental set-up, and how the simulation inputs
are generated.  Then, we explore some crucial system design parameters
that determine both the computational time and quality of our
solutions. Once we have determined suitable values of these parameters,
we will evaluate our approach (and compare with a lower-bound on costs,
and other heuristics) for a range of inputs values.
\vspace{-0.1in}
\subsection{Input Generation}
As described in Section~\ref{sec:model}, our system inputs consist of
$n$ cooperating caches, each cache $i$ has a storage capacity $s_i$,
upload and download usage caps $u_i$ and $d_i$ respectively. Also, there
exist $n_u$ users and $m$ content items (requested and served) in the
system. \cut{As mentioned before, the cooperating caches will be
controlled by a central entity. This entity will collect usage statistics at each
cache over a period of time, and create an estimate of the requests that
will arrive at each cache within the next decision period.} We consider two
different system scenarios: 1) {\em single association, no sharing} (SANS): each
user associates only with one cache and requests some items, but does not
generate or share content with other users (this corresponds to the scenario of
home users downloading  movies, or software updates). To generate this input, we
randomly associate each cache with a user, and the user randomly selects on
average $n_{pop\_items}$ out of available $m$ items; 2) {\em multiple
associations, with sharing} (MAWS): a more general scenario,  in which users on
average associate with $n_{assoc}$ caches, and can also generate and share
$n_{shares}$ pieces of content with $n_{friends}$. The number of associations
reflect the number of caches a user is likely to visit (and spend time) in
the course of a day (e.g. bus-stops near home and work). While generating this
input, our goal was also to preserve a notion of geographical {\em locality
of requests}, i.e. some items are requested more frequently from a small set of
caches, and also {\em locality of social graph}, i.e. friends have a greater
chance of associating with the same caches. To ensure this, we created a mapping
between items to caches, and with a higher likelihood, users requested items
mapped to their associated caches. Similarly to create the social graph,
with high likelihood, a user selected friends from caches with common associations.

{\em Real World Traces:} We also used real-world traces 
(from the CRAWDAD repository\cite{crawdad-wifi}) collected at
Wi-Fi access points distributed across the city of Montreal. These traces 
gave us realistic data on user association patterns with 
multiple public hot-spots. However, due to the limited scale 
of the data (only $\approx 150$ hot-spots), bulk of our experiments 
are performed
with the synthetically generated input as described above. 


{\em Request patterns across epochs:} Another key aspect of the simulated
system is the dynamic nature of the request patterns over time. To capture this
in our simulations, at each epoch we randomly select and expire $5\%$ of items
from the caches. We then randomly select and repopulate the expired items in the cache.

\cut{
{\em Real World Traces:} In our experiments, we also used a publicly available
wifi session trace obtained from the CRAWDAD data repository \cite{}. The trace
consists of user sessions across $140$ hot-spots distributed throughout the city
of Montreal and is collected over a span of three years. However, due to the
limited scale of the data (only $\approx 150$ caches), most of our experiments are performed
with synthetically generated data as mentioned above.
}

Based on these strategies, we associate users with caches, items and friends,
and compute the number of requests for each item at each cache in the system. This
constitutes the input to the CPLEX LP solver. The rounding method (and
other post-processing steps) applied to the resulting solution provide the costs
computed by our approach, and for ease of exposition, we will
subsequently refer to it as $C_{round}$. 

\vspace{-0.1in}
\subsection{Choosing Cluster Size}
\cut{
\begin{figure*}[tbp]
\begin{minipage}[p]{0.3\textwidth}
\centering
\includegraphics[width=\textwidth]{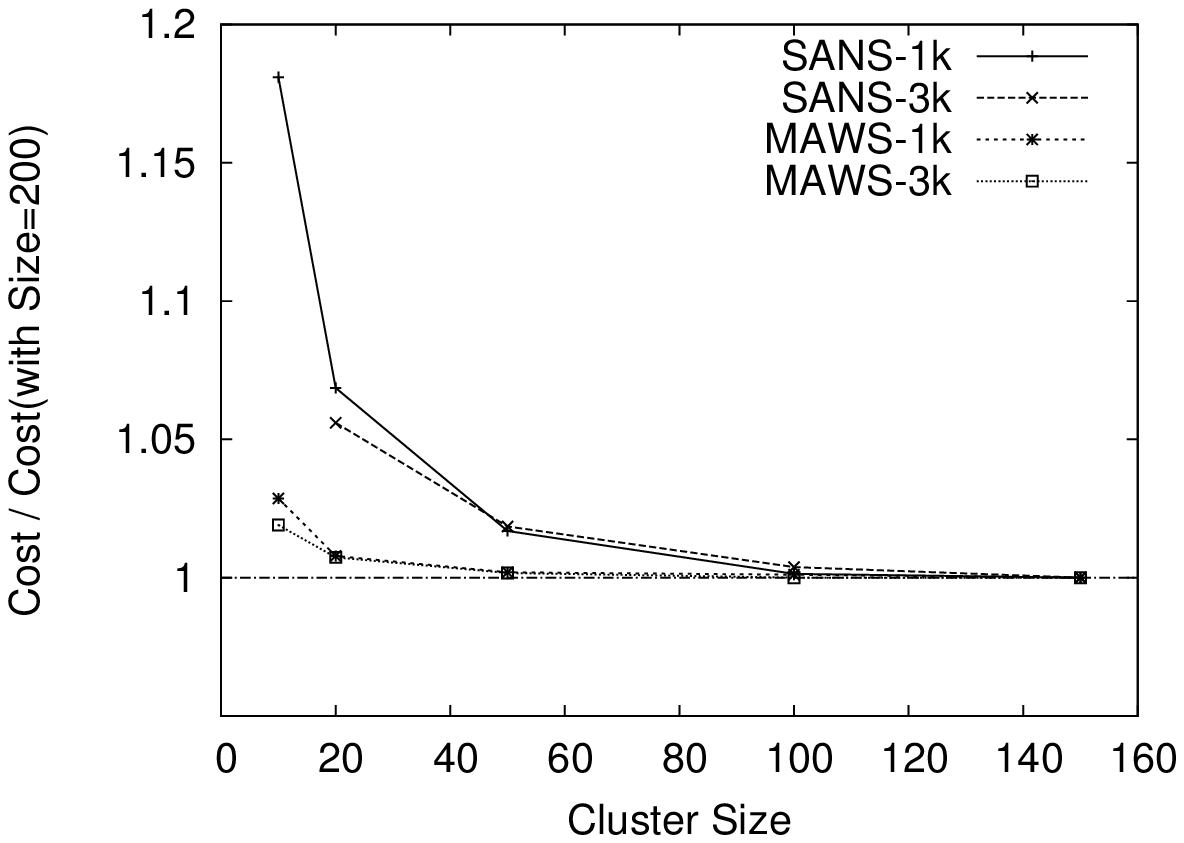}
\vspace{-.2in}
\caption{Costs for different cluster sizes}
\label{fig:cost-vs-clustsize}
\end{minipage}
\begin{minipage}[p]{0.3\textwidth}
\centering
\includegraphics[width=\textwidth]{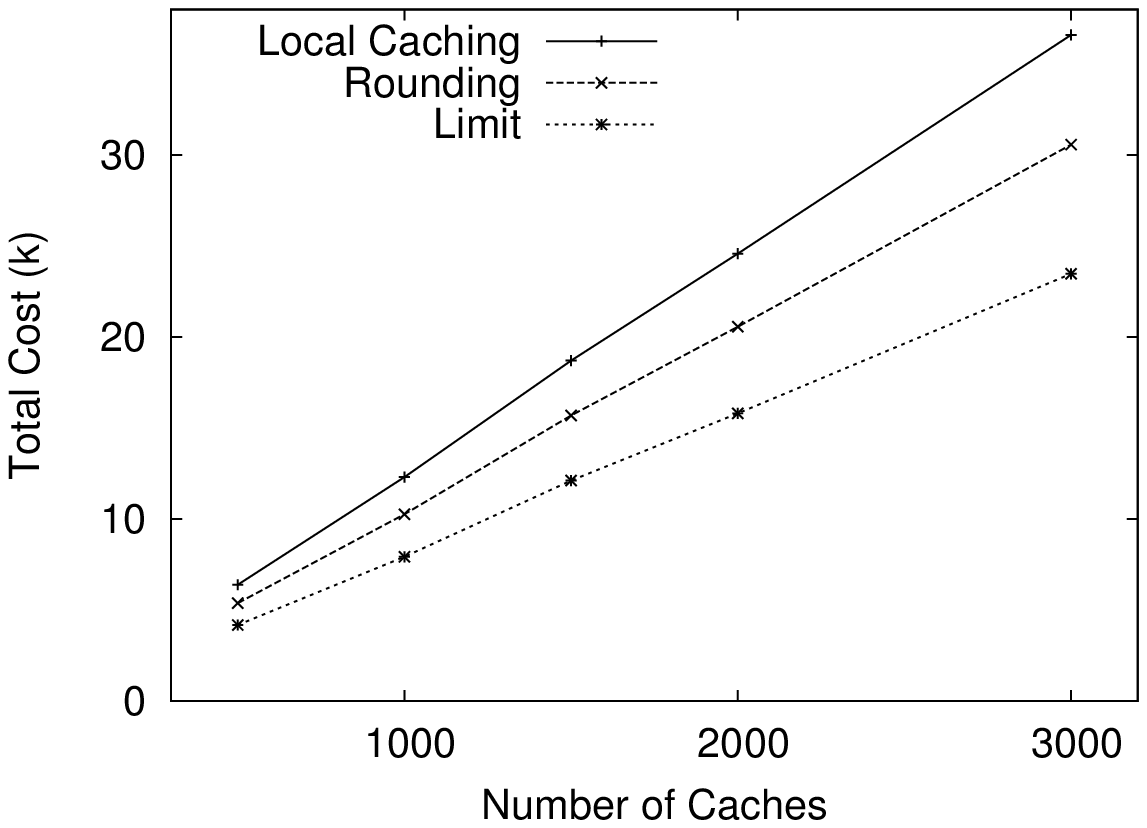}
\vspace{-.2in}
\caption{$C_{round}$ vs. $n$ (MAWS).}
\label{fig:cost-vs-ncaches-mango}
\end{minipage}
\begin{minipage}[p]{0.3\textwidth}
\centering
\includegraphics[width=\textwidth]{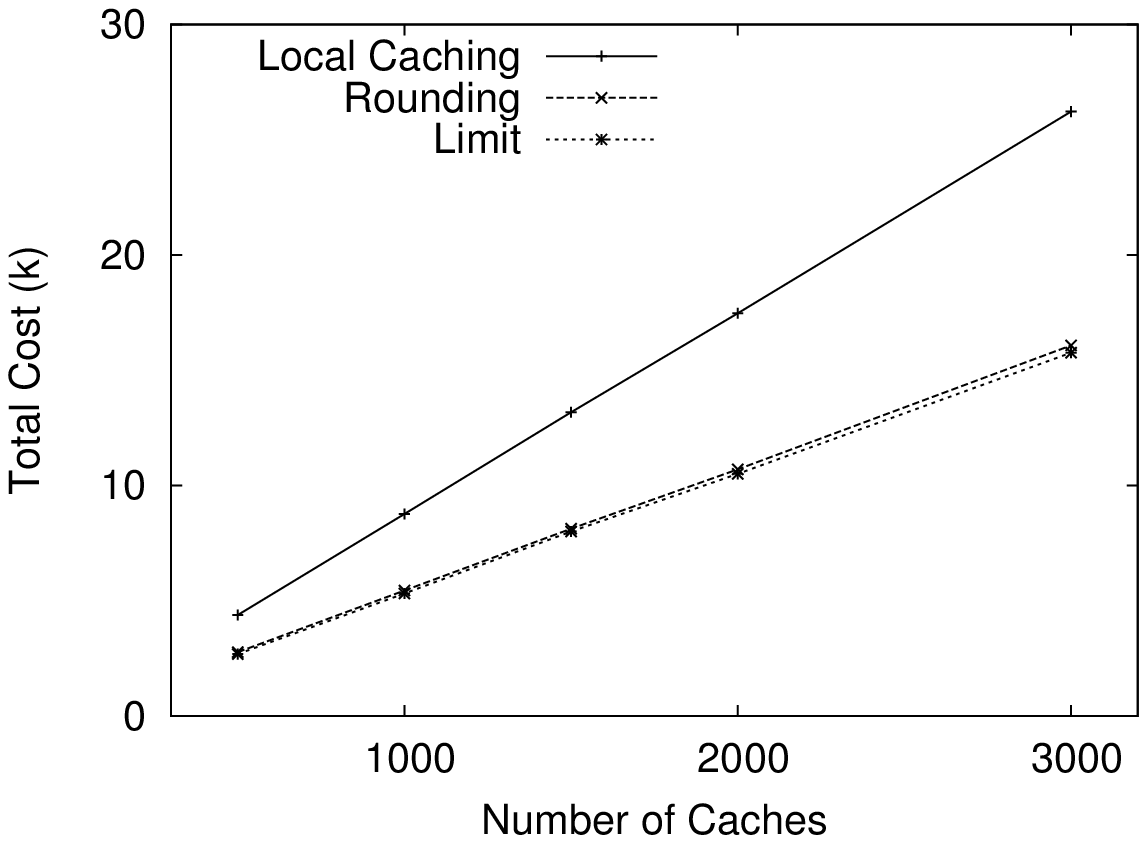}
\vspace{-.2in}
\caption{$C_{round}$ vs. $n$ (SANS).}
\label{fig:cost-vs-ncaches-netflix}
\end{minipage}
\end{figure*}
}

\begin{figure*}[tbp]
\begin{minipage}[p]{0.3\textwidth}
\centering
\includegraphics[width=1.05\textwidth]{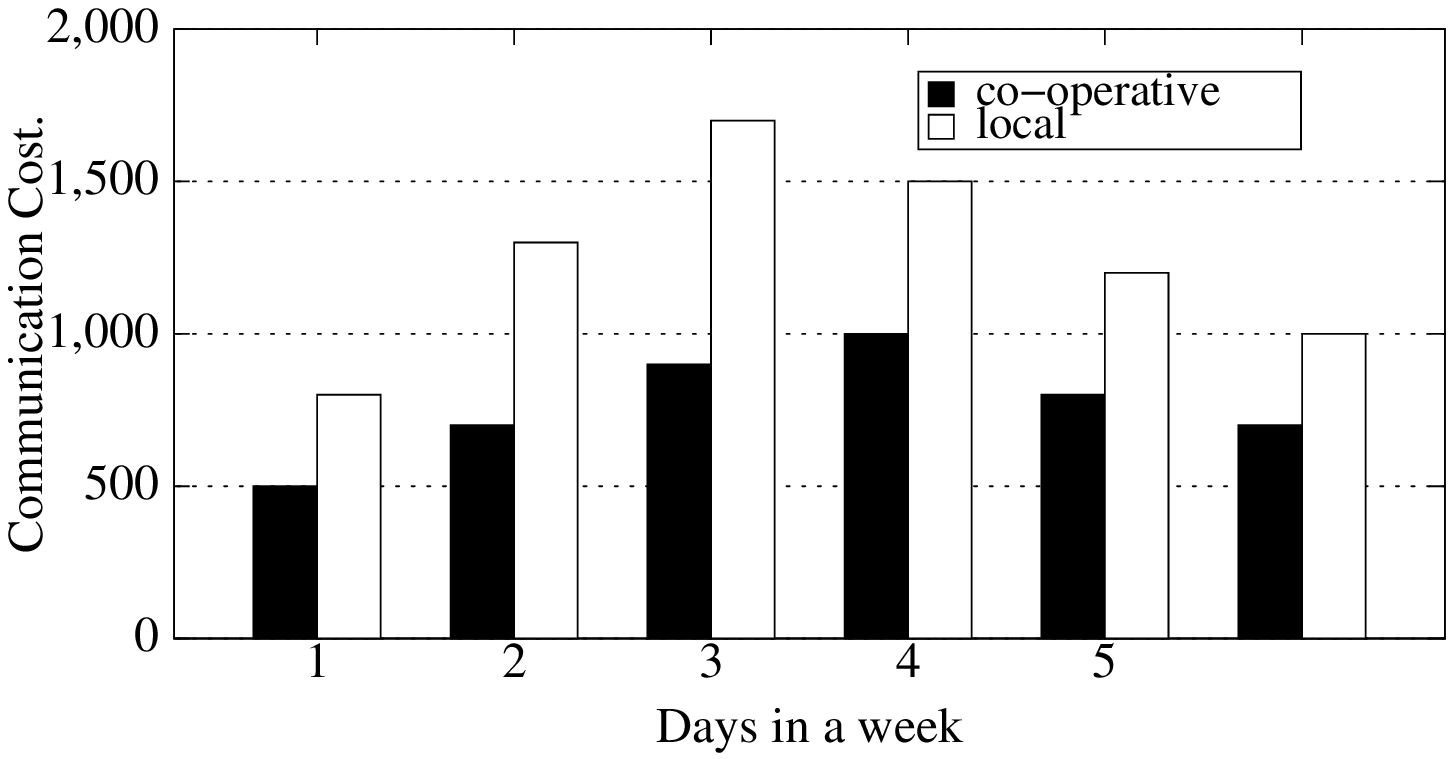}
\vspace{-.2in}
\caption{$C_{round}$ vs. Local Caching on WiFi trace}
\label{fig:wifi}
\end{minipage}
\begin{minipage}[p]{0.3\textwidth}
\centering
\includegraphics[width=\textwidth]{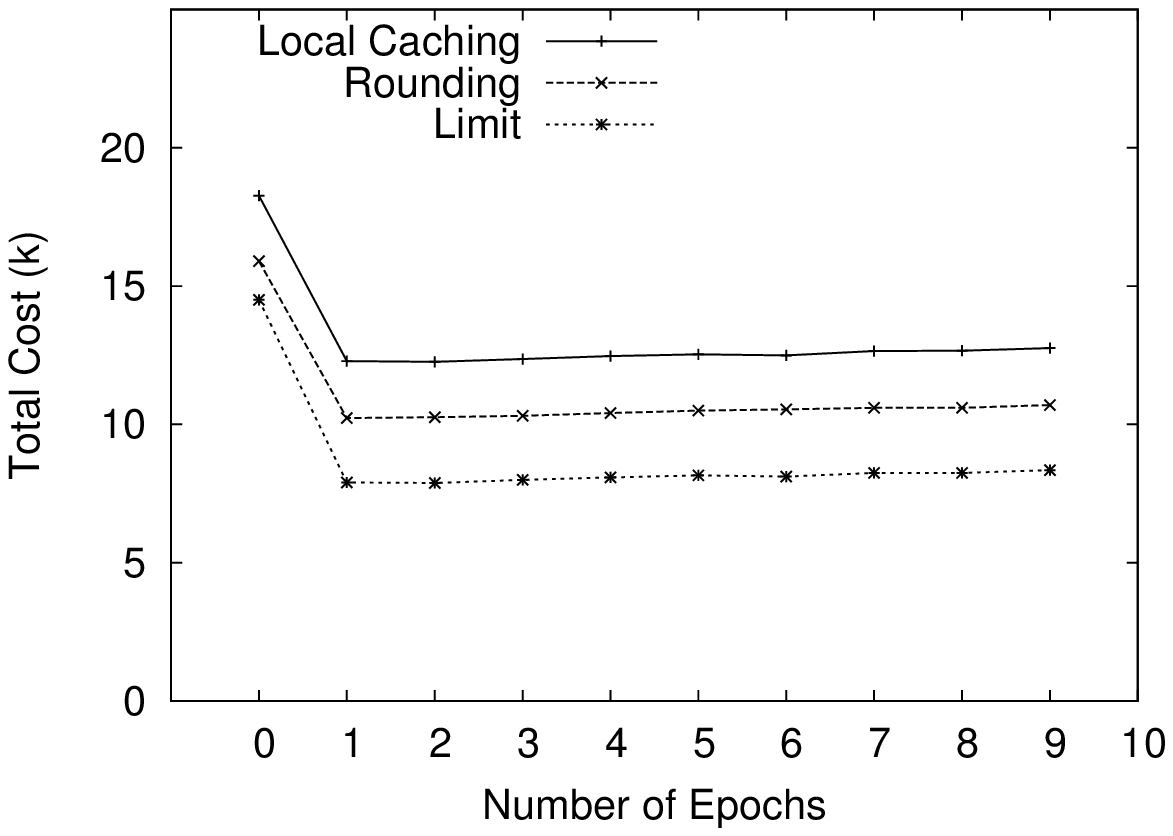}
\vspace{-.2in}
\caption{$C_{round}$ vs. epochs (MAWS).}
\label{fig:cost-vs-nepochs-mango}
\end{minipage}
\begin{minipage}[p]{0.3\textwidth}
\centering
\includegraphics[width=\textwidth]{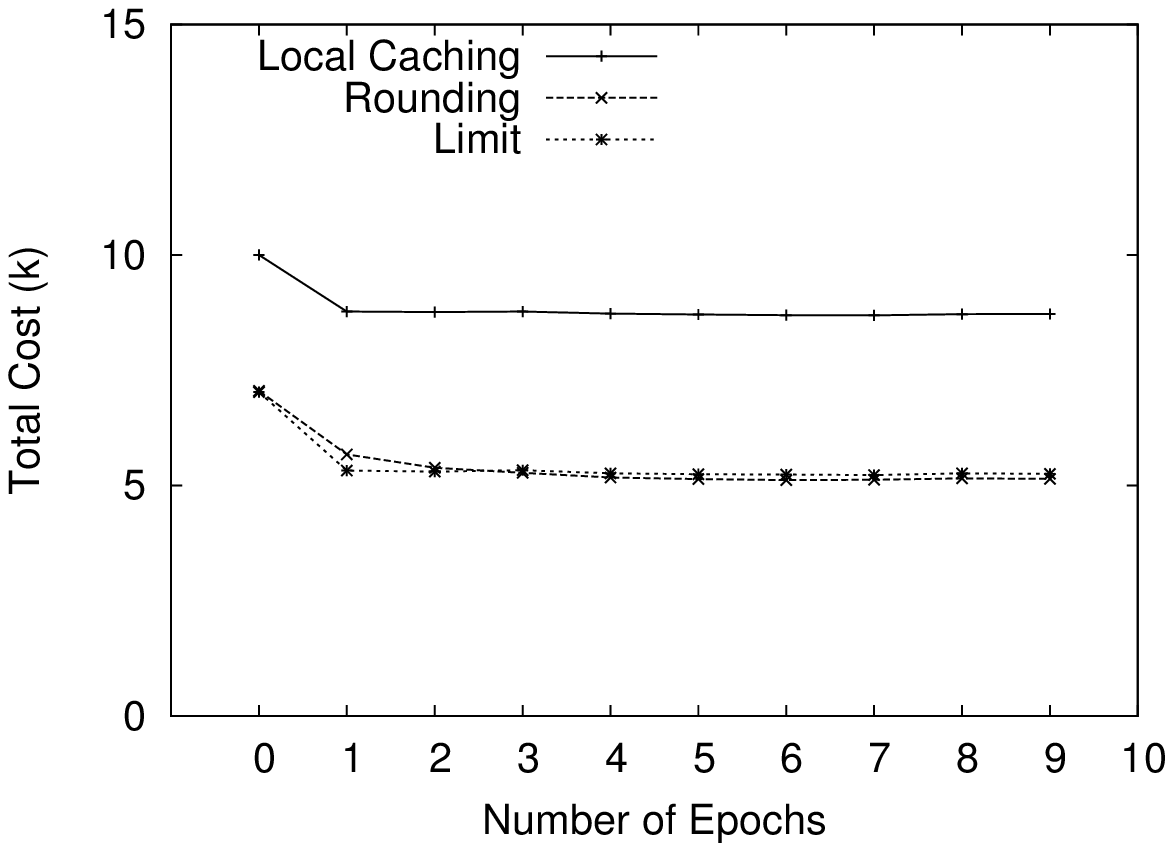}
\vspace{-.2in}
\caption{$C_{round}$ vs. epochs (SANS).}
\label{fig:cost-vs-nepochs-netflix}
\end{minipage}
\end{figure*}


\cut{As mentioned earlier in Section~\ref{sec:ext}, solving the LP for a system with thousands of caches, users and items has a high computational overhead. In Section~\ref{sec:ext}, we introduced an alternate approach (\clust) to reduce the problem size by partitioning caches into clusters, and recombining the results for an approximation of the overall cost of data-placement. }
Our first experiment is designed to determine a reasonable cluster size for the 
clustering procedure.
Intuitively, we expect the output costs from 
clustering to improve when the size of clusters 
is large (with a subsequent increase in execution time).
In Fig.~\ref{fig:cost-vs-clustsize}, we plot the impact of cluster size on the cost of the \round\ algorithm. To evaluate the marginal gains of increasing size, we normalized the costs by dividing it with the cost with cluster size $200$ (as this is the largest cluster size that we could run).
The input values for the SANS case are $n = 1000$, $m = 100$, $u_i$ randomly selected between $1-5$,  $\tau = 1/3$, $n_u = 10000$, and $n_{pop\_items} = 10$. Additionally, for MAWS, $n_{assoc} = 2$ and $n_{shares} = 3$. As can be seen, beyond  cluster sizes  $> 50$ the marginal gains of larger cluster size are insignificant. Moreover, the execution time on this particular input grows $5$ times from a cluster size of $50$ to $100$. Thus, for subsequent experiments, we fix the cluster size to $50$.

\cut{
The second design parameter we explore is $\tau$. As explained earlier, it is expected that increasing the available storage at each cache will improve the performance of the data placement scheme. We explore the ``bank-for-the-buck" of having more storage at each cache. 
We compared $C_{round}$ for both SANS and MAWS for different values of $\tau$. We found there exists $\tau$ at which $C_{round}$ comes extremely close to the lower-bound $C_{lb}$ (within $10$\%), and an increase in the storage offers only marginal benefits. Given this we select $\tau = 1$ for SANS, and $\tau = 3$ for MAWS, because amongst the $\tau$ considered, these values offer the best combination of low costs and storage requirements across different inputs.}

\vspace{-0.1in}
\subsection{Comparison and evaluation over a range of inputs}
Now that we have identified suitable values of design parameters, we now
evaluate and compare the performance of our approach across different
values of $n$, both for the SANS scenario in
Fig.~\ref{fig:cost-vs-ncaches-netflix} and MAWS
in Fig.~\ref{fig:cost-vs-ncaches-mango}. The values for other input
parameters are as before. The results are averaged over $5$ different
input graphs and $3$ epoch periods. The goal of these experiments
is two fold. We quantify the benefits of co-operative caching as proposed in our
approach ({\em Rounding}) compared to when caching is only done locally, at each node
({\em Local-caching}), and find that our co-operative algorithm reduces content delivery costs by 
$20\--30\%$\footnote{Note that in SANS, local caching is not expected to give any savings as there
is at most one request of any item on a cache, hence this experiment is only for the MAWS scenario}. 

Moreover, we also found that our approach drastically reduces network load
at the central server by $55$\% and $27$\% for both the MAWS and SANS
cases respectively. The reduction in delivery costs is also significant in experiments with 
the real-life WiFi traces (Figure~\ref{fig:wifi}), where the {\em Rounding} approach outperforms
{\em Local-caching} by more than $50\%$. These gains are
significant since they indicate a service provider has considerable
buffer to incentivize users to set aside some of their resources
for co-operative caching.



Then we compare our approach with \cut{an alternate {\em Greedy} approach and also with} $C_{lb}$ a lower-bound on the benefits of co-operative caching. The {\em Rounding} approach stays well within $20$\% of the lower bound in the MAWS case, and within $5$\% in the SANS case. \cut{Also the greedy approach has between $10\--30$\% more costs as compared to the rounding solution, however it has significantly fast computational time, in the order of minutes for even the largest input set.}

 In our experiments, we have also computed the blow-ups in the upload and storage limits. Interestingly, we found that on average there is practically \emph{no} blow-up in either limit. Specifically, we found storage is increased by less than $1\%$, while about $5\%$ of upload limit remains unused.
Further, even the worst-case blow-ups of storage and upload limits are $3$ and $4.5$, much less than the worst case factors proved in Theorem~\ref{thm:apx}.

\cut{
\begin{figure*}[tbhp]
\begin{minipage}[p]{0.3\textwidth}
\centering
\includegraphics[width=\textwidth]{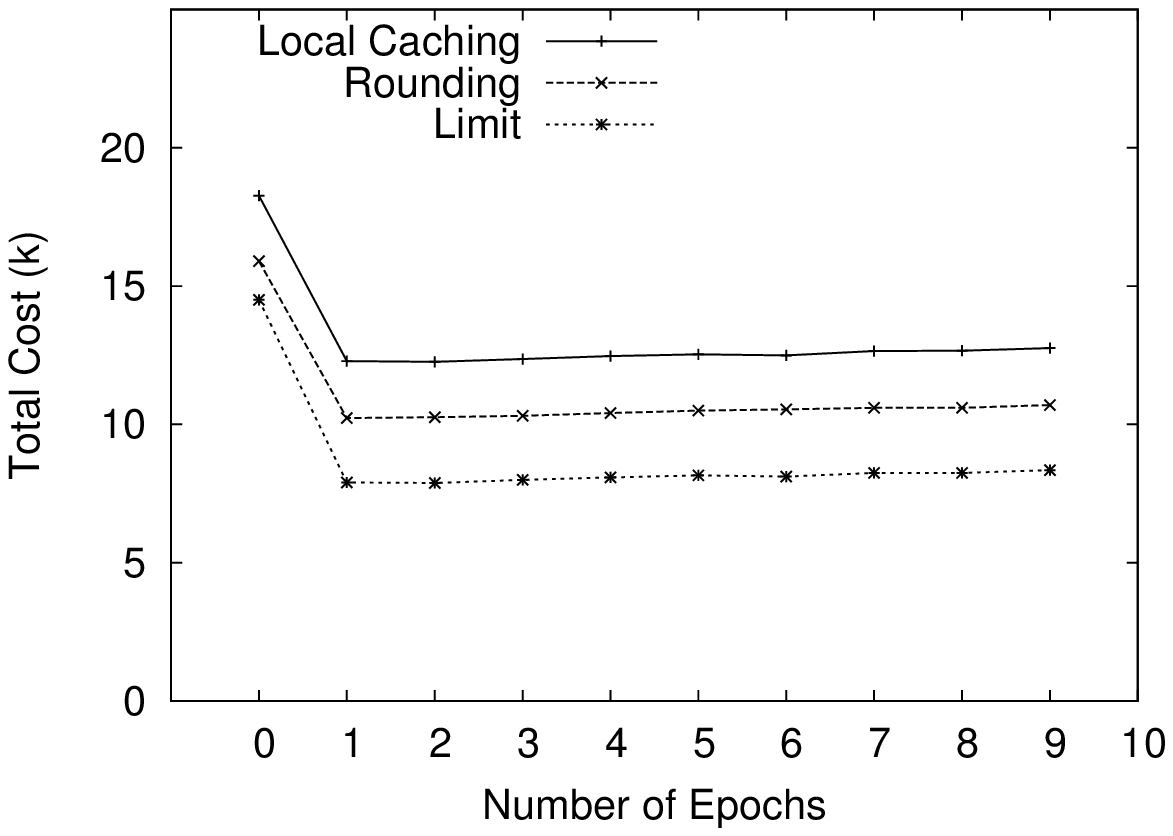}
\vspace{-.2in}
\caption{$C_{round}$ vs. $n$, multiple associations with sharing (MAWS) scenario.} \label{fig:online-mango}
\label{fig:cost-vs-nepochs-mango}
\end{minipage}
\begin{minipage}[p]{0.3\textwidth}
\centering
\includegraphics[width=\textwidth]{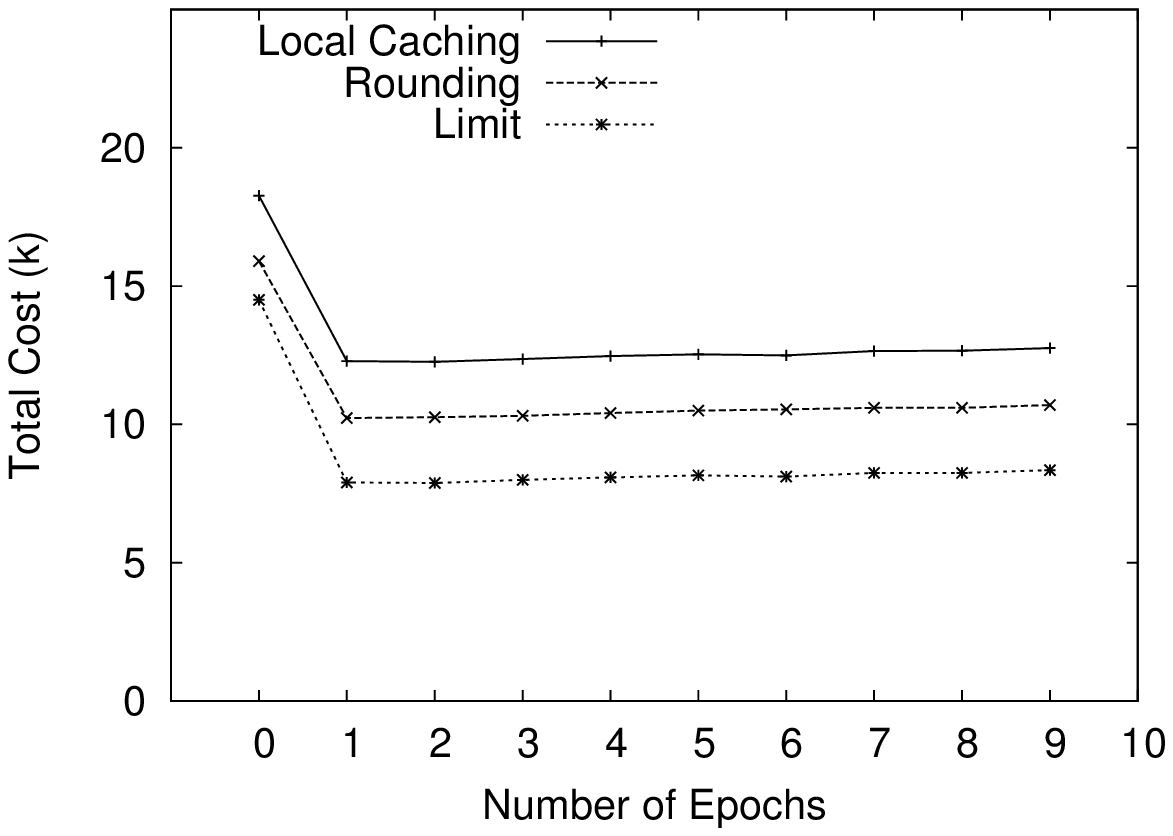}
\vspace{-.2in}
\caption{$C_{round}$ vs. $n$, single association with no sharing (SANS) scenario.} \label{fig:online-netflix}
\label{fig:cost-vs-nepochs-netflix}
\end{minipage}
\begin{minipage}[p]{0.3\textwidth}
\centering
\includegraphics[width=\textwidth]{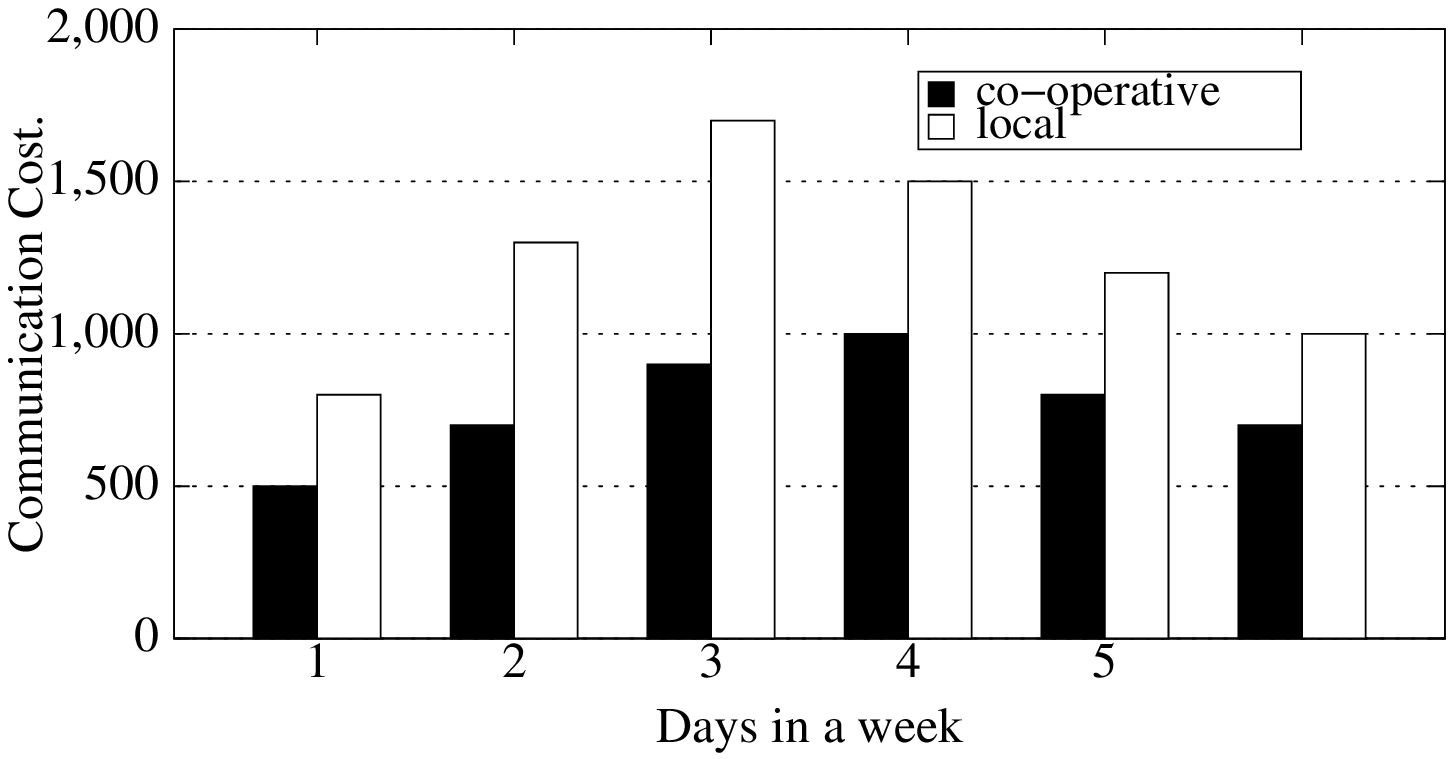} \label{fig:online-wifi}
\vspace{-.2in}
\caption{$C_{round}$ vs Local on Wi-Fi trace}
\end{minipage}
\end{figure*}
}
\subsection{Dynamic request patterns over time}

\cut{Until now we have evaluated the performance of our data-placement 
problem in one snapshot of time. We now proceed to the next stage - 
across requests dynamically changing in time. }

We now study the cost varies over different epochs.
As discussed earlier, 
a prediction of requests arriving at a cache is computed 
across an {\em epoch}. Across epochs, given the current state 
of the caches, and a new set of requests, our algorithm computes 
a data-placement solution that minimizes the cost of serving 
content fetched either from the central server or a co-operative 
cache (as the case maybe). The data-placement costs of 
our {\em Rounding} algorithm (Figure~\ref{fig:cost-vs-nepochs-mango} and Figure~\ref{fig:cost-vs-nepochs-netflix}) 
are compared to that of a {\em Local-caching} scheme that computes the data-placement 
solution by considering what is already present in the caches. 
We make two observations, firstly, both in SANS and MAWS data-sets, the cost of data-placement 
significantly decreases after the first epoch (when the caches are empty), 
and stabilizes in subsequent epochs. Secondly, our scheme outperforms 
the {\em Local-caching} scheme by $20\%$, both in the first and the subsequent epochs. 

To conclude, through a fairly extensive set of simulations 
we have explored the performance of our proposed co-operative 
caching scheme. We have explored the system design space, 
and understood various aspects of our approach such as the 
computational time and blowups. We believe these results 
convincingly indicate that through a careful management of 
co-operative end-nodes one can substantially lower the cost 
of content delivery, while respecting the limitations 
of the end-host network resources.

\section{Conclusions} \label{sec:conclusions}

In this paper we consider a co-operative caching system consisting of end-nodes. Each node brings to the table some spare storage capacity, and spare capacity on their network back-haul links. Keeping in mind the limits at each node, we devise a centralized strategy to efficiently manage the distributed system resources, while reducing the cost on the central node. We believe this work has important implications for any end-node based caching solution, with tiered network access costs.  Moving forward, interesting directions for future work would include 1) predicting requests at a cache based on historical behaviorial patterns of users 2) managing a dynamic system with the cooperating end-nodes coming in and going out of the system 3) a distributed approach towards co-operative caching with similar constraints as studied in this paper.

{
\footnotesize
\bibliographystyle{IEEEtran}
\bibliography{mangocdn}
}

\appendix
\input{appendix}
\end{document}

%% file: appendix.tex
\subsection{Proof of NP-Hardness}
\label{app:np-proof}

\begin{proof}
We reduce a given instance $P(O,W,t)$ of the partition problem into our problem
as follows. The cache placement instance contains four nodes with $C_0$ being
 the server. The node $C_1$ gets all the requests in $R$, s.t. the number of
requests for object $o_k$ is $r_{1k} = 2\cdot w_k$. Nodes $C_2$ and $C_3$ get
no request of their own and only help $C_1$ in serving its request (and hence
reducing the load on the server). The limits and costs on the nodes are as
follows: for the server $u_0=m$, and $d_0=0$; for $C_1$, $s_1=0$, $d_1=2W$, $u_1=0$;
for nodes $C_2$ and $C_3$, $d_2=s_2=t$, $d_3=s_3=m-t$ and $L^u_2=L^u_2=W$.
The extra costs are $\alpha_i=\beta_i=1$ ($\forall i\in[0,3]$).
In other words, node $C_1$ cannot cache any object locally, and has to get them
from one of the other nodes. Further, nodes $C_2$ and $C_3$ can download and
store only $t$ and $n-t$ objects, respectively, before their limits are exhausted.
The limit on the server ensures that if it serves more than $m$ downloads
(one per distinct item) then we start paying cost per download.

We want to prove a $1$-$1$ correspondence between our caching problem and the
partition problem. Specifically, the cost of caching problem is zero (when all
transfers stay within the upload/download limits) if and only if there is a
partition $T$. The if part is true as if there is such a partition $T$, we can
download objects in $T$ in cache $C_2$ and $O-T$ in $C_3$ and then serve all
requests from them, giving zero cost. For the only if part, let us assume that
there is a zero cost solution. Clearly, in this solution all the requests are
served from $C_2$ and $C_3$, since there are at least $2$ requests (as $w_i\geq 1$)
for each object and if any of them is served directly from the server, its upload
limit will be crossed. Further, since $C_2$ and $C_3$ together serve a total of $2W$
requests and each limit is $W$, for a zero cost solution they must each serve
exactly $W$ requests. The set of objects in $C_2$ represents partition.
This completes the proof.
\end{proof}

\subsection{Proof of Theorem~\ref{THM:ONLINE}}
\label{app:online-proof}
\begin{proof}
Observe that the effect of rounding on the first part of the cost
function (i.e. $\sum_{j,k}x_{0jk}$) follows from the analysis 
presented in Lemma~\ref{LEM:ROUNDING}. Therefore, we 
only consider the second part in the cost function. 
Let $OPT_f(Y)$ denote the second part in the cost function corresponding
to the optimum
fractional solution of the data placement problem.
We follow the steps in the rounding procedure and bound the approximation factor of 
$OPT_f(Y)$ at each such
step. 
\cut{
Consider the steps as in previous section. Let $C_x$ be the $x$-part of the cost function, which is $\min \sum_{j,k}  x_{0jk}$ and $C_y$ be the $y$-part of the cost function, which is $\sum_{i} \sum_{k \notin I_i} z_{ik}$. Let us also use $C'_y$ to denote $\sum_i \sum_{k \notin I_i} y_{ik}$. Our rounding procedure remains exactly the same except for Step 4, where we solve GAP($n, s', y, C'_y$) instead of ...

Let us first consider rounding of $y_{ik}$, which happens in steps 1-4. 
Since $\frac{y_{ik}}{y_{ik}+1}$ is increasing function of $y_{ik}$,
Note that the rounding of $y_{ik}$ will chan
analysis of rounding $x_{ijk}$ and effect of it on the cost function and approximation factors does not change. Only rounding of $y_{ik}$ changes
only steps 1, 2 and 4 effects $y_{ik}$ and the
}
\medskip

\noindent {\bf (Step 1)} Since $f(y_{ik})=\frac{y_{ik}}{y_{ik}+1}$ is an 
increasing function of $y_{ik}$, this step increases $OPT_f(Y)$ only 
by a factor of $\frac{f(1)}{f(0)}=1.5$.

\noindent {\bf (Step 2)} This step only decreases $OPT_f(Y)$, again because 
$f(y_{ik})$ is an increasing function of $y_{ik}$.

\noindent {\bf (Step 3)} This step has no effect on $OPT_f(Y)$.

\noindent {\bf (Step 4)} Let us now solve GAP($n, s', y, OPT_f(Y)$).\footnote {Observe that
this GAP instantiation refers to the original rounding technique 
described in Shmoys \etal~\cite{st-soda-93}. This GAP problem has a cost function 
and the rounding 
technique ensures that the cost of the integral solution is at most the cost
of the optimum fractional solution.}
Since $0 \leq y_{ik} \leq 1$ for all $y_{ik}$ going in to GAP and GAP algorithm rounds each 
$y_{ik}$ to at most $2$, the rounded solution satisfies 
$1 \leq y_{ik}+1 \leq 3$ and therfore the rounded solution  
is a $3$-approx to $OPT_f(Y)$.

\noindent {\bf (Step 5 \& 6)} These steps has no effect on $OPT_f(Y)$.

If the minimum cost of the convex optimization problem 
can be expressed as $OPT_f(X)+OPT_f(Y)$,
then the rounding step produces a solution of cost at most $8\cdot OPT_f(X)+ \frac{9}{2}
\cdot OPT_f(Y)$. 
The factors for upload and storage follows 
from the analysis presented in Lemma~\ref{LEM:ROUNDING}.
\end{proof}